\documentclass[12pt,onecolumn,twoside]{IEEEtran}

\usepackage{calc}

\usepackage{multirow}
\usepackage{cite}
\usepackage{graphicx,subfigure}
\usepackage{psfrag}
\usepackage{amsmath,amssymb,amsthm}
\usepackage{color}
\usepackage{tikz} 
\usepackage{bm}
\usepackage{bbm}
\usepackage{verbatim}
\usepackage[T1]{fontenc}

\usepackage{cases}
\usepackage[noend]{algpseudocode}

\usepackage{xcolor}
\usepackage[most]{tcolorbox}
\usepackage{bm}

\usepackage{tabu}
\usepackage{mathtools}
\usepackage{xifthen}
\usepackage{booktabs}
\usepackage{algorithm}
\usepackage{algpseudocode}
\usepackage{arydshln}

\makeatletter
\def\BState{\State\hskip-\ALG@thistlm}
\makeatother

\DeclareMathOperator*{\defeq}{\triangleq}

\newtheorem{theorem}{Theorem}

\newtheorem{lemma}{Lemma}
\newtheorem{definition}{Definition}

\newcommand{\bit}{\begin{itemize}}
\newcommand{\eit}{\end{itemize}}

\newcommand{\bc}{\begin{center}}
\newcommand{\ec}{\end{center}}

\newcommand{\ba}{\begin{array}}
\newcommand{\ea}{\end{array}}

\newcommand{\beq}{\begin{equation}}
\newcommand{\eeq}{\end{equation}}

\newcommand{\beqn}{\begin{equation*}}
\newcommand{\eeqn}{\end{equation*}}

\newcommand{\bean}{\begin{eqnarray*}}
\newcommand{\eean}{\end{eqnarray*}}
\newcommand{\bea}{\begin{eqnarray}}
\newcommand{\eea}{\end{eqnarray}}

\def\hv{\boldsymbol{h}}

\def\wv{\boldsymbol{w}}

\newcommand{\Ac}{{\mathcal A}}
\newcommand{\Bc}{{\mathcal B}}
\newcommand{\Cc}{{\mathcal C}}
\newcommand{\Dc}{{\mathcal D}}
\newcommand{\Ec}{{\mathcal E}}
\newcommand{\Fc}{{\mathcal F}}

\newcommand{\Pc}{{\mathcal P}}

\newcommand{\Uc}{{\mathcal U}}

\newcommand{\T}{{\scriptscriptstyle\mathsf{T}}}

\algnewcommand{\IfThenElse}[3]{  \State \algorithmicif\ #1\ \algorithmicthen\ #2\ \algorithmicelse\ #3}

\newcommand{\non}{\nonumber}

\newcommand\bl[1]{{\color{blue}#1}}

\newcommand{\Pros}{\text{processors}}
\newcommand{\Pro}{\text{processor}}
\newcommand{\PRO}{\text{Processor}}
\newcommand{\Lh}{\ell}
\newcommand{\Me}{\wv}
\newcommand{\Ry}{\mathrm{s}}

\newcommand{\Vr}{\mathrm{v}} 
\newcommand{\gn}{\eta} 
 
\newcommand{\Meg}{\bar{\wv}}

\newcommand{\cb}{c} 

\newcommand{\Ss}{\mathbb{S}}

\newcommand{\Lin}{l}
\newcommand{\Lk}{\mathrm{u}}

\newcommand{\rtuple}{)} 
\newcommand{\ltuple}{(}

\newcommand{\AND}{\land}    
\newcommand{\OR}{\lor}

\newcommand{\MVBAInputMsg}{\wv}

\newcommand{\EncodedSymbol}{y}

\newcommand{\thisnodeindex}{i}                       
\newcommand{\READY}{\text{``}\mathrm{READY}\text{''}}                           
\newcommand{\send}{\textbf{send}}

\newcommand{\Output}{\textbf{output}}                      
\newcommand{\terminate}{\textbf{terminate}}

\newcommand{\wait}{\textbf{wait}}                             
                        
\newcommand{\ECC}{\mathrm{ECC}}

\newcommand{\ECCEnc}{\mathrm{ECCEnc}}   
\newcommand{\ECCDec}{\mathrm{ECCDec}} 
\newcommand{\Alphabet}{\Bc}      
        
\newcommand{\MVBAOutputMsg}{\hat{\wv}}

\newcommand{\INITIAL}{\text{``}\mathrm{INITIAL}\text{''}}

\newcommand{\OECsymbolset}{\mathbb{Z}_{\mathrm{oec}}}   
\newcommand{\Lkset}{\mathbb{U}}   
\newcommand{\SYMBOL}{\text{``}\mathrm{SYMBOL}\text{''}}                        
\newcommand{\SIone}{\text{``}\mathrm{SI1}\text{''}}                         
\newcommand{\SItwo}{\text{``}\mathrm{SI2}\text{''}}

\newcommand{\defaultvalue}{\bot}      
\newcommand{\ECCEncindicator}{I_\mathrm{ecc}}    
      
\newcommand{\Phoneindicator}{I_\mathrm{1}}       
\newcommand{\Phtwoindicator}{I_\mathrm{2}}        
\newcommand{\Phthreeindicator}{I_\mathrm{3}}

\newcommand{\CORRECTSYMBOL}{\text{``}\mathrm{CORRECT}\text{''}}                         
\newcommand{\OECCorrectSymbolSet}{\mathbb{Y}_{\mathrm{oec}}}

\newcommand{\OEC}{\mathrm{OEC}}

\newcommand{\COOL}{\mathrm{COOL}}

\newcommand{\DM}{\mathrm{DM}}

\newcommand{\LEADER}{\text{``}\mathrm{LEAD}\text{''}}                            
 
\newcommand{\networksizen}{n}                                           
\newcommand{\kencode}{k}                                            
\newcommand{\networkfaultsizet}{t}

\newcommand{\successindicator}{\mathrm{s}}

\newcommand{\IDCOOL}{\mathrm{ID}}                                                                       
\newcommand{\OciorCOOL}{\text{OciorCOOL}}           
\newcommand{\BinaryBARound}{\Bc_{\mathrm{r}}}      
\newcommand{\SIPhOne}{\text{``}\mathrm{SI}\text{''}}                            
\newcommand{\SIPhTwo}{\text{``}\mathrm{NewSI}\text{''}}                            
       
\newcommand{\BA}{\mathrm{BA}}

\newcommand{\Input}{\textbf{input}}                      
\newcommand{\OciorBUA}{\text{OciorUA}}         
\newcommand{\BBA}{\mathrm{BBA}}                                            
\newcommand{\BBAoutput}{\Vr^{\star}}

\newcommand{\OciorRBC}{\text{OciorRBC}}

\newcommand{\Byzantineuniqueagreement}{\text{Unique  agreement}}            
\newcommand{\BUA}{\mathrm{UA}}

\newcommand{\HMDM}{\mathrm{HMDM}}

\newcommand{\RBC}{\mathrm{RBC}}                                                                          
\newcommand{\RBCleaderindex}{l}

\newcommand{\OECsymbol}{z}                                                                            
\newcommand{\OECsymbolsetInitial}{\mathbb{Z}_{\mathrm{oec}}}     
\newcommand{\OECSI}{I_\mathrm{oec}}                                                                        
\newcommand{\SIPhtwo}{I_\mathrm{SI2}}     
\newcommand{\OECSIFinal}{I_\mathrm{oecfinal}}                                                                              
      
\newcommand{\VrOutput}{\mathrm{v}_o}

\algdef{SE}[SUBALG]{Indent}{EndIndent}{}{\algorithmicend\ }%
\algtext*{Indent}
\algtext*{EndIndent}

\begin{document}
\sloppy
\title{$\OciorCOOL$: Faster Byzantine Agreement and Reliable Broadcast}
\author{Jinyuan Chen 
}

\maketitle
\pagestyle{headings}

\begin{abstract}
$\COOL$  (Chen'21) is an error-free and deterministic Byzantine agreement  protocol that achieves consensus on an $\Lh$-bit message with a communication complexity of $O(\max\{n\Lh, n t \log t \})$  bits in four phases, given $n\geq 3t + 1$, for a network of $n$ nodes, where up to $t$ nodes may be dishonest.
In this work we show that  $\COOL$ can be optimized by reducing one communication round. The new protocol is called   $\OciorCOOL$. 
Additionally, building on $\OciorCOOL$, we design an optimal reliable broadcast protocol that requires only six communication rounds. 
\end{abstract}


\section{Introduction}

Byzantine agreement (BA) is a fundamental distributed consensus problem introduced around forty years ago \cite{PSL:80}. In this problem, $n$ distributed nodes seek to reach consensus on  an $\Lh$-bit message, where up to $t$ of the nodes may be dishonest. Byzantine agreement, together with its variants such as Byzantine broadcast  (BB) and reliable broadcast ($\RBC$), is  believed to be an essential foundation of distributed systems and cryptography  \cite{PSL:80, LSP:82, ChenDISC:21, Chen:2020arxiv, LCabaISIT:21, ZLC:23, FH:06, LV:11, GP:20, LDK:20, NRSVX:20, Patra:11, CT:05}.

For the multi-valued error-free BA problem, significant efforts have been made to improve performance in terms of communication and round complexities, under the optimal resilience condition $n\geq 3t + 1$  \cite{LV:11, GP:20,LDK:20, NRSVX:20,ChenDISC:21} (see Table~\ref{tb:BA}). In this direction,  Chen designed the  $\COOL$ protocol that  achieves a communication complexity of $O(\max\{n\Lh, n t \log t \})$  bits with four phases, under the optimal resilience of  $n\geq 3t + 1$ \cite{ChenDISC:21}.
In this work, we demonstrate that $\COOL$ can be optimized by eliminating one phase, thereby reducing the number of communication rounds. With fewer communication rounds, the new protocol, called $\OciorCOOL$, is faster than the original $\COOL$ protocol.

$\COOL$ has been used as a building block in other consensus problems, such as BB \cite{ChenDISC:21}, asynchronous BA  \cite{ChenDISC:21},  gradecast \cite{ZLC:23},  validated Byzantine agreement \cite{CDG+:24},  and $\RBC$ \cite{ADD+:22}. 
In this work, building on $\OciorCOOL$, we design an error-free reliable broadcast protocol called  $\OciorRBC$, which requires only six communication rounds and improves upon the $\RBC$ protocol by Alhaddad et al.\cite{ADD+:22}, which requires eight communication rounds (see Table~\ref{tb:BRC}).  The proposed $\OciorCOOL$ can be applied to other consensus problems, such as BB, asynchronous BA,  gradecast, and validated Byzantine agreement to improve the round complexity.

The proposed $\OciorCOOL$ protocol is described in Algorithms~\ref{algm:OciorBUA} and \ref{algm:OciorCOOL}, and its analysis is provided in Section~\ref{sec:OciorCOOL}.
The proposed $\OciorRBC$ protocol is described in Algorithm~\ref{algm:OciorRBC}, and its analysis is provided in Section~\ref{sec:OciorRBC}.  Table~\ref{tb:BA} and Table~\ref{tb:BRC} provide the comparison between the proposed protocols and some other error-free protocols for the BA and $\RBC$ settings, respectively. 
Some definitions and primitives  used in our protocols are provided in the following subsection.

\subsection{Primitives}

 \vspace{.1 in}
    
    \noindent  {\bf Information-theoretic (IT) protocol.} A protocol that guarantees all of the required properties without using any cryptographic assumptions, such as  signatures and hashing, is said to be \emph{information-theoretic secure}.    The proposed protocols are information-theoretic secure.

\vspace{.1 in} 

\noindent  {\bf Error-free protocol.}    A protocol that  that guarantees all of the required properties in \emph{all} executions is said to be \emph{error-free}.    The proposed protocols are error-free.

{\renewcommand{\arraystretch}{1.3}
\begin{table}
\small
\begin{center}
\caption{Comparison between proposed $\OciorCOOL$  and some other  error-free Byzantine agreement protocols.  $\BinaryBARound(1)$ denotes the round complexity of a binary BA. By using the binary BA protocol in \cite{BGP:92,CW:92}, we have $\BinaryBARound(1)= O(t)$.} \label{tb:BA}
      \vspace{-.05 in}
\begin{tabular}{||c||c|c|c|c|c|}
\hline
Protocols & Resilience &   Communication   &  Rounds   &     Error Free & Signature Free     \\ 
\hline
Liang-Vaidya \cite{LV:11}   & $  n   \geq 3t + 1 $ & $O( n \ell + n^4  \sqrt{\ell}+n^6)$  &  $\Omega (\sqrt{\ell}+n^2)$  &  Yes  & Yes  \\
\hline
Ganesh-Patra \cite{GP:20}   & $  n   \geq 3t + 1 $ & $O( n \ell + n^4 )$  &  $O (t)$  &  Yes  & Yes  \\
\hline
Loveless et al.  \cite{LDK:20}   & $  n   \geq 3t + 1 $ & $O( n \ell + n^4 )$   &  $O(t)$  &  Yes  & Yes  \\
\hline
Nayak et al. \cite{NRSVX:20}   & $  n   \geq 3t + 1 $ & $O(n \ell + n^3)$  &  $O(t)$  &  Yes  & Yes  \\
\hline
  Chen \cite{ChenDISC:21}  &    $  n   \geq 3t + 1 $   &  $O(\max\{n\Lh, n t \log t \})$  &   5   + $\BinaryBARound(1)$     &     Yes  &   Yes \\   
\hline
 {\color{blue}  $\OciorCOOL$} &   {\color{blue}$  n   \geq 3t + 1 $}   & {\color{blue}$O(\max\{n\Lh, n t \log t \})$}  &   {\color{blue}4   + $\BinaryBARound(1)$}       &    {\color{blue}Yes}  &  {\color{blue}Yes} \\
\hline
\end{tabular}
\end{center}
\end{table}
}

{\renewcommand{\arraystretch}{1.3}
\begin{table}
\footnotesize  
\begin{center}
\caption{Comparison between proposed $\OciorRBC$  and some other error-free reliable broadcast protocols.   ``Balanced communication'' means that communication overhead is distributed evenly among distributed nodes.} \label{tb:BRC}   
\begin{tabular}{||c||c|c|c|c|c|}
\hline
Protocols & Resilience  &  Communication        &    Rounds   &    Error Free  &   Signature Free   \\ 
\hline
Bracha  \cite{Bracha:87}    &    $n\geq 3t + 1$ &    $O(n^2 |\MVBAInputMsg|)$       &  $O(1)$   &  Yes  & Yes  \\ 
\hline
Patra  \cite{Patra:11}    &    $n\geq 3t + 1$ &    $O(n |\MVBAInputMsg|+  n^4\log n)$     &  $O(1)$   &  Yes  & Yes  \\ 
\hline
Nayak et al.\cite{NRSVX:20}    &    $n\geq 3t + 1$ &    $O(n |\MVBAInputMsg|+  n^3\log n)$      &  $O(1)$   &  Yes  & Yes  \\ 
\hline
Alhaddad et al.\cite{ADD+:22}    &    $n\geq 3t + 1$ &    $O(n |\MVBAInputMsg|+  n^2\log n)$      &  8 (without balanced com.)   &  Yes  & Yes  \\ 
    &      &          &  9 (with balanced com.)   &   &    \\ 
\hline
  {\color{blue} $\OciorRBC$} &         $ \bl{n\geq 3t + 1}$ & {\color{blue} $O(n |\MVBAInputMsg|+  n^2\log n)$  }    &   {\color{blue} $6$ (without balanced com.)}     &    {\color{blue} Yes}     &    {\color{blue}Yes}      \\
  &          &      &   {\color{blue} $7$ (with balanced com.)}     &        &       \\
  \hline
\end{tabular}
\end{center}
\end{table}
}

\vspace{.1 in} 

 \noindent  {\bf Error correction code ($\ECC$).}   An $(n, k)$ error correction coding  scheme  consists of an encoding scheme $\ECCEnc: \Alphabet^{k} \to  \Alphabet^{n}$ and a decoding scheme $\ECCDec: \Alphabet^{n'} \to  \Alphabet^{k}$, where $\Alphabet$ denotes the alphabet of each symbol, for some $n'$. 
 While $[\EncodedSymbol_1,  \EncodedSymbol_2, \cdots, \EncodedSymbol_{n}] \gets \ECCEnc (n,  k, \MVBAInputMsg)$ outputs $n$ encoded symbols, 
$\EncodedSymbol_{j} \gets \ECCEnc_{j}(n,  k, \MVBAInputMsg)$ outputs the $j$th encoded symbol.  
  An $(n, k)$ error correction code can correct up to  $t$ Byzantine errors and simultaneously detect up to $d$ Byzantine errors in $n'$ symbol observations, given the conditions of  $2t+ d +k \leq   n'$  and  $n' \leq  n$.  
 Reed-Solomon (RS) code (cf.~\cite{RS:60}) is one popular error correction code.   The $(n, k)$ RS code is operated over Galois Field $GF(2^c)$ under the constraint of $n \leq 2^c  -1$ (cf.~\cite{RS:60}). Berlekamp-Welch algorithm and Euclid's algorithm are two efficient decoding algorithms for RS code \cite{roth:06, Berlekamp:68, RS:60}.

\vspace{.1 in}

 \noindent  {\bf Online error  correction  ($\OEC$).}     
Online error  correction is a variant of traditional error correction \cite{BCG:93}.   An $(n, k)$ error correction code can correct up to  $t'$ Byzantine errors in $n'$ symbol observations, provided the conditions of  $2t'+ k \leq   n'$  and  $n' \leq  n$. However, in an asynchronous setting, a node might not be able to decode the message with $n'$ symbol observations if  $2t'+ k >   n'$. 
In such a case, the node can wait for one more symbol observation before attempting to decode again. This process repeats until the node successfully decodes the message. By setting the threshold as $n' \geq k + t$, $\OEC$ may perform up to $t$ trials in the worst case before decoding the message.

\begin{definition} [{\bf Byzantine agreement}]
In the Byzantine agreement  protocol, the distributed nodes   seek to reach agreement on a common value. 
The  $\BA$ protocol guarantees  the following properties: 
\begin{itemize}
\item  {\bf Termination:} If all  honest nodes receive their inputs, then every honest node  eventually outputs a value and terminates. 
\item  {\bf Consistency:} If any honest node output a value $\wv$, then every honest node eventually outputs $\wv$.
\item  {\bf Validity:}     If all honest nodes input the same  value $\wv$, then every honest node eventually outputs $\wv$.  
\end{itemize} 
\end{definition}

\begin{definition} [{\bf Reliable broadcast} \cite{Bracha:87}]
 In a reliable broadcast protocol, a leader  inputs a value  and broadcasts it to distributed nodes,   satisfying the following conditions:
\begin{itemize}
\item  {\bf Consistency:} If any two honest nodes output $\wv'$ and $\wv''$, respectively, then  $\wv'=\wv''$.
\item   {\bf Validity:} If the leader is  honest and inputs a value $\wv$, then every honest node eventually outputs $\wv$. 
\item  {\bf Totality:}  If one  honest node outputs a value, then every honest node  eventually outputs a value.          
\end{itemize} 
\end{definition}

\begin{definition} [{\bf Distributed multicast}] \label{def:DM}
 In the problem of distributed multicast ($\DM$), there exits a subset of nodes acting as  senders multicasting the message over $n$ nodes, where up to $t$ nodes could be dishonest. Each node acting as an sender has an input   message.  A protocol is called as a $\DM$ protocol if  the following property is guaranteed:  
\begin{itemize}
\item   {\bf Validity:} If all  honest senders input the same message $\wv$,   every honest node eventually outputs $\wv$.       
\end{itemize} 
{\bf Honest-majority distributed multicast ($\HMDM$):} A $\DM$ problem is called as honest-majority $\DM$ if  at least   $t+1$ senders are honest.  $\HMDM$ was used previously as a building block for $\COOL$ protocol, i.e., Phase 4 of $\COOL$ \cite{Chen:2020arxiv, ChenDISC:21}. 
\end{definition}

\begin{definition} [{\bf  $\Byzantineuniqueagreement$}]  \label{def:BUA}     
$\Byzantineuniqueagreement$ ($\BUA$)   is a variant of  Byzantine agreement  problem operated over $n$ nodes, where up to $t$ nodes may be dishonest.  
In a $\BUA$ protocol, each node inputs an initial value and  seeks to make an  output taking the form as $(\MVBAInputMsg, \successindicator, \Vr)$, where $\successindicator \in \{0,1\}$ is a    success indicator and $\Vr \in \{0,1\}$ is a vote. 
The  $\BUA$ protocol guarantees  the following properties: 
\begin{itemize}
\item  {\bf Unique Agreement:} If any two honest nodes output $(\wv', 1, *)$ and $(\wv'', 1, *)$, respectively, then  $\wv'=\wv''$.
\item  {\bf Majority Unique Agreement:} If any honest node outputs $(\wv, 1, 1)$, then at least $t+1$ honest nodes eventually output  $(\wv, 1, *)$. 
\item   {\bf Validity:} If all honest nodes input the same  value $\wv$,  then  all honest nodes eventually  output $(\wv, 1, 1)$. 
\end{itemize} 
\end{definition}

\section{$\OciorCOOL$}    \label{sec:OciorCOOL}

This proposed $\OciorCOOL$ is a  deterministic and error-free Byzantine agreement  protocol for the synchronous setting.  
$\OciorCOOL$ doesn't rely on any cryptographic assumptions such as signatures or hashing. 
This proposed $\OciorCOOL$ protocol is an improvement on the previous $\COOL$ protocol, using three phases instead of four \cite{Chen:2020arxiv, ChenDISC:21}.

\subsection{Overview of $\OciorCOOL$}    \label{sec:OverviewOciorCOOL}

The proposed  $\OciorCOOL$ is described in Algorithm~\ref{algm:OciorCOOL} and Algorithm~\ref{algm:OciorBUA}. In the following, we provide an overview of the proposed protocol.

\subsubsection{Phases~1 and 2}  The first two phases uses the proposed $\OciorBUA$ algorithm (Algorithm~\ref{algm:OciorBUA}) as a building block (Line~\ref{line:OciorCOOLOciorBUA} of Algorithm~\ref{algm:OciorCOOL}).  $\OciorBUA$ is a $\BUA$ algorithm which ensures that:  1) if any two honest nodes output $(\wv', 1, *)$ and $(\wv'', 1, *)$, respectively, then  $\wv'=\wv''$ (Unique Agreement);  2) if any honest node outputs $(\wv, 1, 1)$, then at least $t+1$ honest nodes eventually output  $(\wv, 1, *)$ (Majority Unique Agreement); and  3) if all honest nodes input the same  value $\wv$,  then  all honest nodes eventually  output $(\wv, 1, 1)$ (Validity).

After delivering outputs from $\OciorBUA$, Node~$i$ makes a vote and runs a binary BA consensus on  the votes.  This ensures sure that all honest nodes make the same decision on whether to terminate at Phase~2 or go to the next phase.

\subsubsection{Phase~3} This phase uses distributed multicast as a building block.  This phase ensures that the encoded symbols from the honest nodes can be calibrated using majority rule such that the symbols are encoded from the same message. In this way, all honest nodes with success indicators of zero can output the same decoded message.

\begin{algorithm}
\caption{$\OciorBUA$  protocol with identifier $\IDCOOL$. Code is shown for $S_{\thisnodeindex}, \thisnodeindex \in [\networksizen]$}  \label{algm:OciorBUA}
\begin{algorithmic}[1]
\vspace{5pt}    
\footnotesize

\State  $\Input$   $\Me_{i}$ 
\State  Initially set   $\kencode \gets   \bigl \lfloor   \frac{ \networkfaultsizet  }{5 } \bigr\rfloor    +1 ;  \Me^{(i)}  \gets \Me_{i}$      \quad \quad\quad\quad\quad \quad\quad\quad\quad \quad\quad\quad\quad  \quad \emph{// set the initial value of $\Me^{(i)}$}
\State  $[y_1^{(i)}, y_2^{(i)}, \cdots, y_{n}^{(i)}] \gets \ECCEnc(\networksizen,  \kencode, \Me_{i})$  \quad \quad \quad \quad\quad\quad\quad\quad\quad\quad\quad\quad  \emph{// ECC encoding }

\vspace{5pt} 
  
\Statex {\bf \emph{Phase~1}}

\State   $\send$  $\ltuple   \SYMBOL, \IDCOOL,  (y_j^{(i)}, y_i^{(i)}) \rtuple$ to $S_j$,  $\forall j \in  [\networksizen]$  \  \quad\quad \quad \quad\quad\quad \quad \quad \emph{// exchange coded  symbols; $[\networksizen]:=\{1,2,\cdots, \networksizen\}$}

\For     {$j =1:n$}    \quad     \quad\quad \quad \quad \quad \quad \quad\quad \quad \quad  \quad \quad \quad \quad \quad \quad \quad \quad  \quad \quad \quad \quad \emph{//set link indicator}

    \IfThenElse {$ (y_i^{(j)}, y_j^{(j)}) = (y_i^{(i)}, y_j^{(i)})$}  {$\Lk_i (j) \gets 1$ } {$\Lk_i (j) \gets 0$ }   \label{line:OciorBUAlinkset111}

   \EndFor

     \IfThenElse {$ \sum_{j=1}^n \Lk_i (j)  \geq  n - t $}  {$\Ry_i \gets 1$ } {$\Ry_i  \gets 0$; $\Me^{(i)} \gets \defaultvalue$}   \label{line:OciorBUASIOne}   \quad \quad \quad\quad \  \emph{//set success indicator}

  \State   $\send$   $\ltuple   \SIPhOne, \IDCOOL,  \Ry_i \rtuple$     to all nodes    \  \quad\quad\quad \quad  \quad\quad\quad\quad\quad \quad  \quad\quad\quad\quad\quad \quad  \emph{//  exchange success indicators}
 
 \State  set $ \Ss_1  =   \{ j:  \Ry_{j}=1,   j \in [1:n ]\}$ and $\Ss_0  =   \{ j:  \Ry_{j}=0,   j \in [1:n ]\}$, based on received success indicators  $ \{\Ry_{j}\}_ {j=1}^{n}$.

   \vspace{5pt} 
\Statex  {\bf \emph{Phase~2}}

 \If {$\Ry_i  =1$}    \label{line:OciorBUAPh2Begin} 

    \State set   $\Lk_i (j) \gets 0,  \forall j \in \Ss_0$    \label{line:OciorBUAmskerror}     \quad\quad\quad  \quad\quad\quad\quad \quad\quad\quad\quad \quad\quad\quad \quad  \quad  \quad\quad \emph{//  mask identified errors}      
    
              \If {$\sum_{j=1}^n \Lk_i (j)  < n - t$}            
         
           	 \State set  $\Ry_i  \gets 0$; $\Me^{(i)} \gets \defaultvalue$           	    \label{line:OciorBUASI2}   \   \quad \quad   \quad \quad \quad \quad \quad \quad  \quad \quad \quad \quad \quad \quad  \quad \quad \quad \quad \emph{//update success indicator}
            	 \State  $\send$   $\ltuple   \SIPhTwo, \IDCOOL,  \Ry_i \rtuple$     to all nodes           \quad    \quad\quad \quad  \quad\quad \quad  \quad\quad \quad  \quad\quad \quad\emph{//  exchange updated success indicators}                     
              \EndIf                   \label{line:OciorBUAPh2End} 
\EndIf

\State update  $\Ry_j \gets \Ry$ if receiving  message  $\ltuple   \SIPhTwo, \IDCOOL,  \Ry \rtuple$ from $S_j, \forall j\in \Ss_1$     \quad   \emph{// update success indicators} 

 \State update   $\Ss_0$ and $\Ss_1$ based on the updated success indicators   $ \{\Ry_{j}\}_ {j}$     \quad\quad   \quad\quad     \emph{// update $\Ss_0$ and $\Ss_1$}

     \IfThenElse {$|\Ss_1|  \geq 2t+1$}  {$\Vr_i \gets 1$} {$\Vr_i \gets 0$}   \label{line:OciorCOOLVote}   \quad\quad\quad\quad\quad \quad \quad\quad\quad\quad\quad \quad\quad  \quad     \emph{// set the vote value}

 \State     $\Output$  $[\Me^{(i)}, \Ry_{i},   \Vr_i,  \Ss_0, \Ss_1,  [y_1^{(1)}, y_2^{(2)}, \cdots, y_{n}^{(n)}], [y_i^{(1)}, y_i^{(2)}, \cdots, y_{i}^{(n)}]]$  	   \quad\quad        \emph{// $\Me^{(i)}, \Ry_{i},   \Vr_i$ are three $\BUA$ outputs}

\end{algorithmic}
\end{algorithm}

\begin{algorithm}
\caption{$\OciorCOOL$     protocol for BA with identifier $\IDCOOL$. Code is shown for $S_{\thisnodeindex}, \thisnodeindex \in [\networksizen]$}  \label{algm:OciorCOOL}
\begin{algorithmic}[1]
\vspace{5pt}    
\footnotesize

\State  $\Input$    a non-empty value  $\Me_{i}$ 		
\State  Initially set   $\kencode \gets   \bigl \lfloor   \frac{ \networkfaultsizet  }{5 } \bigr\rfloor    +1$  
\Statex

\Statex {\bf \emph{Phase~1 and Phase~2}}

\State  $[\Me^{(i)}, \Ry_{i},   \Vr_i,  \Ss_0, \Ss_1,  [y_1^{(1)}, y_2^{(2)}, \cdots, y_{n}^{(n)}], [y_i^{(1)}, y_i^{(2)}, \cdots, y_{i}^{(n)}]] \gets \OciorBUA[ \IDCOOL  ](\Me_{i})$          \label{line:OciorCOOLOciorBUA}    \quad \quad    \emph{//      $\BUA$   with two phases; see Algorithm~\ref{algm:OciorBUA}   }

  	\State  $\BBAoutput\gets \BBA[\IDCOOL](\Vr_i )$               \quad   \quad     \emph{//  $\BBA$ is a binary BA  consensus on  $n$ votes $\{\Vr_1, \Vr_2, \cdots, \Vr_n\}$,  by using    protocol from \cite{BGP:92,CW:92}  }

	\If {   $\BBAoutput=0$}    \quad\quad\quad\quad\quad\quad\quad\quad\quad\quad\quad \quad\quad\quad\quad\quad  \emph{// if the output of binary BA is $0$,  set $\Me^{(i)}$ as a default value}  
		\State   $\Output$  $\defaultvalue$ and $\terminate$ 	   \label{line:OciorCOOLterminateA} 
	\Else
		\State  go to next phase	  	   \label{line:OciorCOOLGo2Ph3} 
	\EndIf

 \Statex

 \Statex {\bf \emph{Phase~3}}
    
 \If {$ \Ry_{i} = 0$}         \label{line:OciorCOOLPh3Begin}     \quad\quad  \quad  \quad\quad   \quad  \quad\quad  \quad  \quad\quad   \quad   \quad\quad   \quad  \quad\quad   \quad  \quad\quad   \quad    \quad\quad  \quad  \quad  \quad\quad   \quad   \quad \quad    \emph{//  $\HMDM$ algorithm with one phase    }
 
\State update  $y_i^{(i)}     \leftarrow  \text{Majority}( \{y_i^{(j)}:   j \in  \Ss_1 \})$     \label{line:OciorCOOLPh3A}    \quad\quad   \quad  \quad\quad\quad\quad   \quad \quad\quad   \quad   \quad\quad   \quad  \quad\quad   \quad    \quad  \emph{// update its coded symbol with majority rule}

\State   $\send$   $\ltuple \CORRECTSYMBOL, \IDCOOL, y_i^{(i)}  \rtuple$  to  the nodes in $\Ss_0$     \label{line:OciorCOOLPh3B}  \quad\quad   \quad   \quad\quad   \quad    \quad\quad   \quad  \quad\quad   \quad   \quad  \quad  \emph{//  broadcast updated symbol} 
	
\State update  $y_j^{(j)}\gets y$ if receiving  message  $\ltuple \CORRECTSYMBOL, \IDCOOL, y   \rtuple$ from $S_j, \forall j\in \Ss_0$       \label{line:OciorCOOLPh3C} \quad\quad   \emph{// update coded symbol}

\State  $\Me^{(i)} \gets \ECCDec(\networksizen,  \kencode, [y_1^{(1)}, y_2^{(2)}, \cdots, y_{n}^{(n)}]  )$  \label{line:OciorCOOLPh3End}    \quad \quad \quad \quad\quad \quad \quad \quad\quad\quad\quad\quad\quad\quad\quad\quad  \emph{// ECC decoding  with updated symbols}

 \EndIf

\State  $\Output$  $\Me^{(i)}$ and $\terminate$ 	       \label{line:OciorCOOLterminateB}
  
\end{algorithmic}
\end{algorithm}

\subsection{Analysis of $\OciorCOOL$}    \label{sec:AnalysisOciorCOOL}

This proposed $\OciorCOOL$ uses three phases, while the previous  $\COOL$ protocol uses four phases \cite{Chen:2020arxiv, ChenDISC:21}.  
We will prove that, even with less number of phases,  $\OciorCOOL$   still guarantees the termination, validity, and consistency properties. 
This proof will follow the original definitions in \cite{Chen:2020arxiv, ChenDISC:21} and will use some results in \cite{Chen:2020arxiv, ChenDISC:21}.

For the ease of notation, we use $\Ry_i^{[p]}$ to denote the value of $\Ry_i$ updated in Phase~$p$, and  use $\Lk_i^{[p]} (j)$   to denote the values of $\Lk_i (j)$ updated in Phase~$p$, for $p\in \{1,2\}$. 
 $\Fc$  is defined as  the set of  indices of  all  dishonest nodes.  \emph{In the analysis} we just focus on the case with  $| \Fc |   = t$. It is worth noting that,   the easier case with  $| \Fc |   = t'$ for $t' < t$   is indistinguishable from the case with $| \Fc |   = t$ in which $t- t'$ out of $t$ dishonest nodes act normally like honest nodes. Therefore, if a protocol  guarantees the  termination, validity, and consistency properties in the extreme case with  $| \Fc |   = t$, it   also guarantees those properties in the easier case with $| \Fc | < t$.

We define some groups of honest nodes as
  \begin{align}
 \Ac_{\Lin} \defeq &  \{  i:    \Me_i =  \Meg_{\Lin},  \  i \notin  \Fc  , \ i \in [1:n]\}, \quad   \Lin \in [1 : \gn]     \label{eq:Aell00}   \\
\Ac_{\Lin}^{[p]} \defeq&   \{  i:  \Ry_i^{[p]} =1, \Me_i =  \Meg_{\Lin},  \  i \notin  \Fc, \ i \in [1:n]\}, \quad   \Lin \in [1 : \gn^{[p]}], \quad  p\in \{1,2\}      \label{eq:Aell}   \\
\Bc^{[p]} \defeq  & \{  i:  \Ry_i^{[p]} =0, \  i \notin  \Fc, \ i \in [1:n] \}, \quad  p\in \{1,2\}    \label{eq:Bdef01} 
 \end{align} 
 for  some different non-empty $\ell$-bit  values $\Meg_{1}, \Meg_{2}, \cdots, \Meg_{\gn}$ and some non-negative integers $\gn, \gn^{[1]},\gn^{[2]}$  such that $ \gn^{[2]} \leq \gn^{[1]} \leq \gn$.  
Group $\Ac_{\Lin}$ (and Group $\Ac_{\Lin}^{[p]}$) can be divided into some possibly overlapping sub-groups  defined as 
    \begin{align}
 \Ac_{\Lin,j} \defeq   & \{  i:  \   i\in  \Ac_{\Lin},   \   \hv_i^\T  \Meg_{\Lin}  = \hv_i^\T  \Meg_j \} , \quad  j\neq \Lin ,   \  j, \Lin \in [1: \gn]   \label{eq:Alj}  \\     
 \Ac_{\Lin,\Lin} \defeq  &  \Ac_{\Lin}   \setminus  \{\cup_{j=1, j\neq \Lin}^{\gn}\Ac_{\Lin,j}\}   , \quad \quad \quad\quad \Lin \in [1: \gn]   \label{eq:All}       \\
 \Ac_{\Lin,j}^{[p]} \defeq   & \{  i:  \   i\in  \Ac_{\Lin}^{[p]},   \   \hv_i^\T  \Meg_{\Lin}  = \hv_i^\T  \Meg_j \} , \quad  j\neq \Lin ,   \  j, \Lin  \in [1: \gn^{[p]}], \quad  p\in \{1,2\}   \label{eq:Alj11}  \\     
 \Ac_{\Lin,\Lin}^{[p]} \defeq  &  \Ac_{\Lin}^{[p]}   \setminus  \{\cup_{j=1, j\neq \Lin}^{\gn^{[p]}}\Ac_{\Lin,j}^{[p]}\}   , \quad\quad\quad\quad  \Lin \in [1: \gn^{[p]}], \quad  p\in \{1,2\}   \label{eq:All11}       
 \end{align}  
 where $ \hv_i$ is the encoding vector of error correction code such that the $i$th encoded symbol is computed as $\EncodedSymbol_i =  \hv_i^\T \MVBAInputMsg$, given the input  vector $\MVBAInputMsg$, for $i\in [1:n]$.

The main results of  $\OciorCOOL$ are summarized in the following Theorems~\ref{thm:OCOOLterminate}-\ref{thm:OciorCOOLPerformance}.   
Theorems~\ref{thm:OCOOLterminate}-\ref{thm:OCOOLconsistency} reveals that, given $n\geq 3t+1$, the  termination, validity and consistency  conditions are all satisfied in all executions (\emph{error-free}).  
Theorems~\ref{thm:OCOOLterminate}-\ref{thm:OCOOLconsistency} hold true without using any cryptographic assumptions (\emph{signature-free}).  Furthermore, Theorems~\ref{thm:OCOOLterminate}-\ref{thm:OCOOLconsistency}  hold true even if the adversary  has unbounded computational power (\emph{information-theoretic secure}).

 \begin{theorem}  [Termination]  \label{thm:OCOOLterminate}
Given $n\geq 3t+1$,  if all  honest nodes receive their inputs, then  every  honest node  eventually  outputs a message and terminates   in $\OciorCOOL$. 
\end{theorem}
\begin{proof}
In $\OciorCOOL$, if all  honest nodes receive their inputs,   all  honest nodes  eventually output  messages and terminate together in Line~\ref{line:OciorCOOLterminateA} of Phase~2  or   terminate together in Line~\ref{line:OciorCOOLterminateB} of Phase~3 in Algorithm~\ref{algm:OciorCOOL}. 
 \end{proof}

 \begin{theorem}  [Validity]  \label{thm:OCOOLvalidity}
Given $n\geq 3t+1$,   if  all  honest nodes input the same  value $\wv$, then  in $\OciorCOOL$  every honest node eventually outputs $\wv$.  
\end{theorem}
\begin{proof}
If all honest nodes  input the same  value $\wv$,  then in Phase~1 each honest node eventually sets its success indicator  as $1$ (see Line~\ref{line:OciorBUASIOne} of Algorithm~\ref{algm:OciorBUA}); and then in Phase~2 each honest node eventually keeps its success indicator  as $1$ (see Lines~\ref{line:OciorBUAPh2Begin}-\ref{line:OciorBUAPh2End} of Algorithm~\ref{algm:OciorBUA}). Thus, each honest node eventually goes to Phase 3 (see Line~\ref{line:OciorCOOLGo2Ph3} of Algorithm~\ref{algm:OciorCOOL})  and then directly  jumps  to Line~\ref{line:OciorCOOLterminateB} of Algorithm~\ref{algm:OciorCOOL}. Therefore, if  all  honest nodes input the same  value $\wv$, then every honest node eventually outputs $\wv$.    
\end{proof}

\begin{theorem}  [Consistency]   \label{thm:OCOOLconsistency}
Given $n\geq 3t+1$,  all  honest nodes  eventually reach  the same agreement in $\OciorCOOL$. 
\end{theorem}
\begin{proof}    
 In  $\OciorCOOL$, if the output of  binary BA ($\BBA$) is $0$, then every honest node outputs  the default value $\defaultvalue$  (see Line~\ref{line:OciorCOOLterminateA} of Algorithm~\ref{algm:OciorCOOL}), satisfying the consistency condition.  
 In the following, we will focus on the case where the output of $\BBA$ is $1$.  
 
 From Lemma~\ref{lm:ABAoutput1},  if the output of  $\BBA$    is $1$, then at least $t+1$ honest nodes have sent out success indicators as $1$ in Phase~2.  
 Furthermore, from  Lemma~\ref{lm:uniquegroup}, it holds true that $\gn^{[2]} \leq 1$, i.e.,  all honest nodes that have sent out success indicators as $1$ in Phase~2 should belong to the same group $\Ac_1^{[2]}$, if any,  where $\Ac_1^{[2]} \defeq   \{  i:  \Ry_i^{[2]} =1, \Me_i =  \Meg_{1},  \  i \notin  \Fc, \ i \in [1:n]\}$   for some $\Meg_{1}$  (see \eqref{eq:Aell}). 
 By combining the results of Lemma~\ref{lm:ABAoutput1} and Lemma~\ref{lm:uniquegroup},   if the output of  $\BBA$    is $1$, then the following conclusions are true 
    \begin{align}
   \gn^{[2]} &= 1      \label{eq:A1g2tConA}     \\
   |\Ac_1^{[2]}|  &\geq    t+1    \label{eq:A1g2tConB}   \\    
       \Me_i &=\Meg_{1} , \quad  \forall i \in  \Ac_1^{[2]} . \label{eq:A1g2tConC} 
 \end{align}
If the output of  $\BBA$    is $1$, and given the conclusions in \eqref{eq:A1g2tConA}-\eqref{eq:A1g2tConC}, then from Lemma~\ref{lm:OciorCOOLHMDM} it is guaranteed that  every honest node eventually outputs the same value $\Meg_{1}$ in Phase~$3$.  
 \end{proof}

\begin{theorem}  [Communication, Round, and Resilience]   \label{thm:OciorCOOLPerformance}
The proposed $\OciorCOOL$  is an error-free signature-free information-theoretic-secure  BA  protocol that achieves the consensus on an $\Lh$-bit message with optimal resilience of $n\geq  3t+1$,  asymptotically optimal round complexity of $O(t)$ rounds, and asymptotically  optimal communication complexity of  $O(\max\{n\Lh, n t \log t \})$  bits, simultaneously.  
\end{theorem}
\begin{proof}
Theorems~\ref{thm:OCOOLterminate}-\ref{thm:OCOOLconsistency} reveals that, given $n\geq 3t+1$, the  termination, validity and consistency  conditions are all satisfied in all executions  (\emph{error-free}) in $\OciorCOOL$.  
The round complexity of $\OciorCOOL$ is dominated by that of the binary BA algorithm, which is $O(t)$ rounds.
The communication complexity of  $\OciorCOOL$ is $O(\max\{n\Lh, n t \log t \})$, similar to that of the $\COOL$ protocol  \cite{Chen:2020arxiv, ChenDISC:21}.
\end{proof}

\subsection{Some lemmas} \label{sec:OciorCOOLsomelemma}

Below we provide some lemmas used in our proofs. Note that some lemmas are directly from \cite{Chen:2020arxiv}.

\begin{lemma}    \label{lm:ABAoutput1}
In $\OciorCOOL$,   if the output of the  binary BA  is $1$, then at least $t+1$ honest nodes have sent out success indicators as $1$ in Phase~2. 
\end{lemma}
 \begin{proof}
 If the output of $\BBA$ is $1$, then at least one   honest node has voted  $1$ in Line~\ref{line:OciorCOOLVote}. Otherwise, $\BBA$ would deliver an output of $0$. 
 When one honest node has voted as $\Vr_i= 1$, it means that this node has seen $|\Ss_1|  \geq 2t+1$  (see Line~\ref{line:OciorCOOLGo2Ph3}), which reveals that at least $t+1$ honest nodes have sent out success indicators as ones in Phase~2, where $\Ss_1$ denotes the indexes of nodes who sent   their success indicators as $1$. 
  \end{proof}

\begin{lemma}    \label{lm:OciorCOOLHMDM}
In $\OciorCOOL$,  if the output of  $\BBA$    is $1$, and given the conclusions in \eqref{eq:A1g2tConA}-\eqref{eq:A1g2tConC}, then every honest node eventually outputs the same value $\Meg_{1}$ in Phase~$3$. 
\end{lemma}
 \begin{proof}
If the output of  $\BBA$    is $1$, then all honest nodes go to Phase~$3$ (see Line~\ref{line:OciorCOOLGo2Ph3} of Algorithm~\ref{algm:OciorCOOL}).  In this case,  all nodes within  $\Ac_1^{[2]}$ directly   jumps  to Line~\ref{line:OciorCOOLterminateB} of Algorithm~\ref{algm:OciorCOOL} and output $\Meg_{1}$. 
With the conclusions in \eqref{eq:A1g2tConA}-\eqref{eq:A1g2tConC}, it can be shown that all honest nodes outside   $\Ac_1^{[2]}$ will output $\Meg_{1}$ as well,  thanks to the honest-majority distributed multicast ($\HMDM$)  protocol of Phase~$3$ (see Lines~\ref{line:OciorCOOLPh3Begin}-\ref{line:OciorCOOLPh3End} of Algorithm~\ref{algm:OciorCOOL}). 
Phase~$3$ is simply an honest-majority distributed multicast protocol. In an honest-majority distributed multicast protocol defined in  Definition~\ref{def:DM},  if  all  honest senders input the same message $\Meg_{1}$ and   at least   $t+1$ senders are honest, then  every honest node eventually outputs $\Meg_{1}$.

Specifically, in Line~\ref{line:OciorCOOLPh3A} of Algorithm~\ref{algm:OciorCOOL},   each honest node with  success indicator being $0$, e.g.,  Node~$i$ with $ \Ry_{i}^{[2]} = 0$,  updates the value of $y_i^{(i)}$  as  $y_i^{(i)}   \leftarrow  \text{Majority}( \{y_i^{(j)}:   j \in \Ss_1\})   = \hv_i^\T   \Meg_{1}$ based on the majority rule, due to the conclusions $\Ac_1^{[2]} \subseteq \Ss_1$  and $ |\Ac_1^{[2]}|  >   |\Fc|$ (see conclusions in \eqref{eq:A1g2tConA}-\eqref{eq:A1g2tConC}). 
After this step,  for any honest Node~$i$, the value of $y_i^{(i)}$  becomes $y_i^{(i)} = \hv_i^\T   \Meg_{1}$ that  is encoded with $\Meg_{1}$. 
Then,  in Line~\ref{line:OciorCOOLPh3B},   each Node~$i$ with $ \Ry_{i}^{[2]} = 0$  sends the  $\ltuple \CORRECTSYMBOL, \IDCOOL, y_i^{(i)}  \rtuple$  to the nodes within $\Ss_0$, where $\Ss_0  =   \{ j:  \Ry_{j}=0,   j \in [1:n ]\}$.  
After the step in Line~\ref{line:OciorCOOLPh3C},    each honest node within  $\Ss_0$  updates  $y_j^{(j)}$ if receiving  message  $\ltuple \CORRECTSYMBOL, \IDCOOL, y_j^{(j)} \rtuple$ from $S_j, \forall j\in \Ss_0$.     
After this step, for each honest node within  $\Ss_0$,   it is guaranteed that each  symbol   $y_j^{(j)}$ sent from an honest node is encoded with $\Meg_{1}$, i.e., $y_j^{(j)} = \hv_j^\T   \Meg_{1}$, for any $j\notin \Fc$.  
Then, in Line~\ref{line:OciorCOOLPh3End},    each honest node within  $\Ss_0$  decodes a message based on the updated symbols $\{y_1^{(1)}, y_2^{(2)}, \cdots, y_n^{(n)}\}$, where  at least $n-t$ symbols are encoded with  $\Meg_{1}$.  Since error correction code can correct up to $t$ errors, each  each honest Node~$i$ within  $\Ss_0$  should decode the message as $\Meg_{1}$, given that  the number of  symbols that are not encoded with the message $\Meg_{1}$ is no more than $t$.  Thus, in this case every honest node eventually outputs the same message $\Meg_{1}$. 
 \end{proof}

  \begin{lemma}    \label{lm:uniquegroup}
For the proposed  $\OciorCOOL$ with $n\geq 3t+1$, it holds true that $\gn^{[2]} \leq 1$.
\end{lemma}
\begin{proof}
At first, from Lemma~\ref{lm:eta2bound2} it is concluded that   $\gn^{[2]} \leq 2$.
Next, we   argue that the case of $\gn^{[2]} = 2$ does not exist.   Let us assume that $\gn^{[2]} = 2$. 
Under the assumption of  $\gn^{[2]} = 2$,  it holds true that $\gn^{[1]} =   2$ (see  Lemma~\ref{lm:eta22212}). However, if $\gn^{[1]} = 2$  then it implies that $\gn^{[2]} \leq   1$ (see  Lemma~\ref{lm:eta1231}), which contradicts the assumption of  $\gn^{[2]} = 2$.
Therefore, the case of $\gn^{[2]} = 2$ does not exist, which, together with the result $\gn^{[2]} \leq 2$, concludes that $\gn^{[2]} \leq 1$.
\end{proof}

 \begin{lemma} \cite[Lemma 7]{Chen:2020arxiv}   \label{lm:sizeboundMatrix} 
For $\gn \geq \gn^{[1]} \geq 2$, the following inequalities  hold true 
    \begin{align}
  |\Ac_{\Lin,j}| + |\Ac_{j,\Lin}|  <  &  k,    \quad \forall   j \neq \Lin,  \  j, \Lin \in [1:\gn]    \label{eq: Aljb00}  \\
    |\Ac_{\Lin,j}^{[1]}| + |\Ac_{j,\Lin}^{[1]}|  <  &  k,  \quad \forall   j \neq \Lin,  \  j, \Lin \in [1:\gn^{[1]}]      \label{eq: Aljb01}  
 \end{align} 
 where $k$ is a parameter of  $(n,k)$ error correction code, which is set here as $k=       \lfloor   \frac{ t  }{5 }  \rfloor    +1$.
  \end{lemma}

\begin{lemma}  \cite[Lemma 11]{Chen:2020arxiv}  \label{lm:eta2bound2}
For the proposed  $\OciorCOOL$ with $n\geq 3t+1$, it holds true that $\gn^{[2]} \leq 2$.
\end{lemma}

\begin{lemma}  \cite[Lemma 13]{Chen:2020arxiv}  \label{lm:eta22212}
For the proposed  $\OciorCOOL$ with $n\geq 3t+1$,  if $\gn^{[2]} = 2$,  then it holds true that $\gn^{[1]} =   2$.   
\end{lemma}

\begin{lemma}    \label{lm:eta1231}
For the proposed  $\OciorCOOL$ with $n\geq 3t + 1$, if $\gn^{[1]} = 2$  then it holds true that $\gn^{[2]} \leq   1$.
\end{lemma}
 \begin{proof}
We will consider the assumption of $\gn^{[1]}= 2$. Given this assumption  the definition in \eqref{eq:Aell00}-\eqref{eq:Bdef01} implies that 
 \begin{align}
 \Ac_{1}^{[1]} =&   \{  i:  \Ry_i^{[1]} =1, \Me_i =  \Meg_{1},  \  i \notin  \Fc, \ i \in [1:n]\}       \label{eq:Aellbig10455}  \\ 
  \Ac_{2}^{[1]} =&   \{  i:  \Ry_i^{[1]} =1, \Me_i =  \Meg_{2},  \  i \notin  \Fc, \ i \in [1:n]\}      \label{eq:Aellbig1045522}    \\ 
\Ac_{1,2}^{[1]} =   & \{  i:  \   i\in  \Ac_{1}^{[1]},   \   \hv_i^\T  \Meg_{1}  = \hv_i^\T  \Meg_2 \}    \label{eq:Alj2995}  \\
\Ac_{1,1}^{[1]} =   & \Ac_{1}^{[1]} \setminus \Ac_{1,2}^{[1]} = \{  i:  \   i\in  \Ac_{1}^{[1]},   \   \hv_i^\T  \Meg_{1}  \neq  \hv_i^\T  \Meg_2 \}  \label{eq:All2955251} \\
\Ac_{2,1}^{[1]} =   & \{  i:  \   i\in  \Ac_{2}^{[1]},   \   \hv_i^\T  \Meg_{2}  = \hv_i^\T  \Meg_1 \}   \label{eq:Alj299535}  \\
\Ac_{2,2}^{[1]} =   & \Ac_{2}^{[1]} \setminus \Ac_{2,1}^{[1]} = \{  i:  \   i\in  \Ac_{2}^{[1]},   \   \hv_i^\T  \Meg_{2}  \neq  \hv_i^\T  \Meg_1 \}  \label{eq:All29559386}  \\
\Bc^{[1]} =  & \{  i:  \Ry_i^{[1]} =0, \  i \notin  \Fc, \ i \in [1:n] \}=  \{  i:    i \in  [1:n], \ i\notin  \Fc \cup  \Ac_{1}^{[1]} \cup  \Ac_{2}^{[1]}  \}   .  \label{eq:BdefB1}         
 \end{align} 
Since $|\Ac_{1}|+|\Ac_{2}|= n-|\Fc|-\sum_{\Lin=3}^{\gn} |\Ac_{\Lin}|$,  it is true that at least one of the following cases is satisfied: 
\begin{align}
\text{Case~1:} \quad  |\Ac_{2}| \leq  &  \frac{n-|\Fc|-\sum_{\Lin=3}^{\gn} |\Ac_{\Lin}|}{2}   \label{eq:case1b} 
 \end{align}
 
\begin{align}
\text{Case~2:} \quad |\Ac_{1}| \leq &\frac{n-|\Fc|-\sum_{\Lin=3}^{\gn} |\Ac_{\Lin}|}{2}  .   \label{eq:case2a}  
 \end{align}

\subsubsection{Analysis for Case~1}   
We will first consider Case~1 and prove that $|\Ac_{2}^{[2]}| =0$ under this case. 
Let us define $\Uc_{i}^{[p]}$ as a set of links that are matched with Node~$i$ at Phase~$p$, that is, 
\begin{align}
\Uc_{i}^{[p]} :=  & \{  j:  \Lk_{i}^{[p]} (j)  =1, j \in [1:n]\},  \quad  \text{for}  \quad i\in [1:n],  p\in \{1,2\}.       
\end{align}
 Then, for any $i\in \Ac_{2}^{[1]}$, the size of $\Uc_{i}^{[2]}$ can be bounded as 
 \begin{align}  
 |\Uc_{i}^{[2]}| = & \sum_{ j \in [1:n]}   \Lk_{i}^{[2]} (j)        \label{eq:link124} \\
 = & \sum_{ j \in [1:n]\setminus \Bc^{[1]}}   \Lk_{i}^{[2]} (j)        \label{eq:link4252} \\
  = & \sum_{ j \in [1:n]\setminus \{\Bc^{[1]}\cup \Ac_{1,1}^{[1]}\}}   \Lk_{i}^{[2]} (j)        \label{eq:link82525} \\
    = & \sum_{ j \in  \{\Ac_{1,2}^{[1]}\cup \Ac_{2,1}^{[1]}\cup \Ac_{2,2}^{[1]} \cup\Fc \}}   \Lk_{i}^{[2]} (j)        \label{eq:link7256} \\
  \leq & |\Ac_{1,2}^{[1]}|  + |\Ac_{2,1}^{[1]}|  + | \Ac_{2,2}^{[1]}|+|\Fc|              \label{eq:link6254} 
\end{align}
where \eqref{eq:link4252} stems from the fact that $\Ry_j^{[1]} =0$ for any $j\in \Bc^{[1]}$ (see \eqref{eq:BdefB1}), which suggests that $\Lk_{i}^{[2]} (j)  =0$ for any  $i\in \Ac_{2}^{[1]}$ (see Line~\ref{line:OciorBUAmskerror} of Algorithm~\ref{algm:OciorBUA}); 
\eqref{eq:link82525} results from the identity that  $\hv_j^\T  \Meg_{1}  \neq  \hv_j^\T  \Meg_2 $ for any  $j\in \Ac_{1,1}^{[1]}$  (see \eqref{eq:All2955251}), which implies that   $\Lk_{i}^{[1]} (j) =\Lk_{i}^{[2]} (j)  =0$      for any $j\in \Ac_{1,1}^{[1]}$ and $i\in \Ac_{2}^{[1]}$  (see Line~\ref{line:OciorBUAlinkset111}   of Algorithm~\ref{algm:OciorBUA});
\eqref{eq:link7256} is from the fact that $[1:n]= \Ac_{1}^{[1]} \cup \Ac_{2}^{[1]}\cup \Bc^{[1]} \cup \Fc = \Ac_{1,1}^{[1]} \cup \Ac_{1,2}^{[1]} \cup\Ac_{2,1}^{[1]}\cup\Ac_{2,2}^{[1]}\cup \Bc^{[1]} \cup \Fc$.

From Lemma~\ref{lm:sizeboundMatrix}, the summation of $|\Ac_{1,2}^{[1]}|  + |\Ac_{2,1}^{[1]}|$ in \eqref{eq:link6254} can be bounded as 
 \begin{align}  
   |\Ac_{1,2}^{[1]}|  + |\Ac_{2,1}^{[1]}|   \leq k-1  .         \label{eq:link3235252} 
\end{align}
From Lemma~\ref{lm:boundA22}, the term $|\Ac_{2,2}^{[1]}|$ in \eqref{eq:link6254} can be upper bounded by 
 \begin{align}
|\Ac_{2,2}^{[1]}| \leq  2(k-1).   \label{eq:A22bound23435} 
\end{align}
At this point, for any $i\in \Ac_{2}^{[1]}$, the size of $\Uc_{i}^{[2]}$ can be bounded as 
 \begin{align}  
 |\Uc_{i}^{[2]}|   \leq & \  |\Ac_{1,2}^{[1]}|  + |\Ac_{2,1}^{[1]}|  + | \Ac_{2,2}^{[1]}|+|\Fc|    \label{eq:A22bound1133} \\
  \leq &\  k-1  + 2(k-1)       +|\Fc|   \label{eq:A22bound8255}  \\ 
  \leq & \ 3(k- 1)  + t   \non  \\ 
    = & \ 3( \lfloor    t /5   \rfloor    +1 -1)   + t  \label{eq:A22bound65464}  \\ 
    \leq &\  2 t   \non  \\ 
     <  & \  n-t  \label{eq:A22bound5366}  
\end{align} 
where  \eqref{eq:A22bound1133} is from    \eqref{eq:link6254};   \eqref{eq:A22bound8255} is from Lemma~\ref{lm:sizeboundMatrix} and Lemma~\ref{lm:boundA22}; \eqref{eq:A22bound65464} uses the value of the parameter $k$, i.e., $k= \lfloor    t /5   \rfloor    +1$; and the last inequality uses the condition of $n\geq 3t+1>3t$.  
With the result in \eqref{eq:A22bound5366}, it is concluded that  Node~$i$, for any $i\in \Ac_{2}^{[1]}$,  sets $\Ry_i^{[2]} =0$ at Phase~2 due to the derived result $|\Uc_{i}^{[2]}| <n-t$ (see Line~\ref{line:OciorBUASI2} of Algorithm~\ref{algm:OciorBUA}). Therefore, it can be concluded that $\Ac_{2}^{[1]}$ is in the list of $\Ss_0$, that is, 
  \begin{align}
\Ac_{2}^{[1]}  \subseteq \Ss_0  \label{eq:case1A2S0} 
\end{align}
  at the end of Phase~2.  In other words,   there exists \emph{at most}  1  group of honest nodes who  input   the same   message  and set their success indicators as ones  at the end of Phase~2, while  the  honest nodes outside this group  set their success indicators as zeros  at the end of Phase~2, that is, $\gn^{[2]} \leq 1$, for Case~1.

\subsubsection{Analysis for Case~2}  
By interchange the roles of $\Ac_{1}$ and  $\Ac_{2}$, one can easily follow the proof for Case~1 and  show that 
   \begin{align}
\Ac_{1}^{[1]}  \subseteq \Ss_0  \label{eq:case1A1S0} 
\end{align}
for Case~2. 
Then it completes the proof of this lemma.  
\end{proof}

\begin{lemma}    \label{lm:boundA22}
Given the condition of $|\Ac_{i}| \leq    \frac{n-|\Fc|-\sum_{\Lin=3}^{\gn} |\Ac_{\Lin}|}{2}$, and for $\gn\geq 2$, the following conclusion is true
\begin{align}
|\Ac_{i,i}^{[1]}| \leq  2(k-1)   \label{eq:A22bound91588} 
\end{align}
for $i\in \{1,2\}$.    
\end{lemma}

\begin{proof}
Without loss of generality, we will  focus on the proof of  $|\Ac_{i,i}^{[1]}| \leq  2(k-1) $   for the case with $i=2$, under the condition of $|\Ac_{i}| \leq    \frac{n-|\Fc|-\sum_{\Lin=3}^{\gn} |\Ac_{\Lin}|}{2}$. The proof for the  case with $i=1$ is similar and thus omitted here.

At first let us consider the case with $\Lin \geq 3$, or equivalently,  $\sum_{\Lin=3}^{\gn} |\Ac_{\Lin}| >0$.  
We will at first argue that $|\Ac_{2,2}^{[1]}|$ is upper bounded by 
 \begin{align}
|\Ac_{2,2}^{[1]}| \leq   \frac{(k-1) \cdot \sum_{\Lin=3}^{\gn} |\Ac_{\Lin}|}{n-t-|\Fc|-|\Ac_{2}| }.    \label{eq:A22bound2145} 
\end{align} 
From the definitions in  \eqref{eq:Aellbig1045522} and \eqref{eq:All29559386},  it is true that  $\Ry_i^{[1]} =1,  \forall i\in \Ac_{2,2}^{[1]}$.   We recall that,    Node~$i$ needs to see 
 \begin{align}
\sum_{ j \in [1:n]}   \Lk_{i}^{[1]} (j)  \geq n-t,\quad     \forall i\in \Ac_{2,2}^{[1]}    \label{eq:A22condition66556} 
\end{align}
 in order to set $\Ry_i^{[1]} =1$ at Phase~1 (see    Line~\ref{line:OciorBUASIOne}  of Algorithm~\ref{algm:OciorBUA}).  
The condition in \eqref{eq:A22condition66556} implies that 
\begin{align}
\sum_{ j \in [1:n]\setminus \{\Fc\cup \Ac_{2}\}}   \Lk_{i}^{[1]} (j)  \geq n-t -|\Fc|-|\Ac_{2}|,\quad     \forall i\in \Ac_{2,2}^{[1]}   .  \label{eq:A22condition327255} 
\end{align}
Furthermore, it is true  that $\Lk_{i}^{[1]} (j) =0$ for any  $i\in \Ac_{2,2}^{[1]}, j \in \Ac_{1}$, due to  identity of  $\hv_i^\T  \Meg_{1}  \neq  \hv_i^\T  \Meg_2 $ for any  $i\in \Ac_{2,2}^{[1]}$  (see \eqref{eq:All29559386}). 
Then,   the condition in \eqref{eq:A22condition327255} can be modified as 
 \begin{align}
\sum_{ j \in [1:n]\setminus \{\Fc\cup \Ac_{2}\cup \Ac_{1}\}}   \Lk_{i}^{[1]} (j)  \geq n-t -|\Fc|-|\Ac_{2}|,\quad     \forall i\in \Ac_{2,2}^{[1]}    . \label{eq:A22condition32525222} 
\end{align}
Since $[1:n]\setminus \{\Fc\cup \Ac_{2}\cup \Ac_{1}\} =\cup_{\Lin=3}^{\gn}\Ac_{\Lin}$,  the condition in \eqref{eq:A22condition32525222} can be expressed as  
 \begin{align}
\sum_{ j \in\cup_{\Lin=3}^{\gn}\Ac_{\Lin}}   \Lk_{i}^{[1]} (j)  \geq n-t -|\Fc|-|\Ac_{2}|,\quad     \forall i\in \Ac_{2,2}^{[1]} ,   \label{eq:A22condition28366} 
\end{align}
which also gives the following bound 
 \begin{align}
\sum_{ i \in \Ac_{2,2}^{[1]}}  \sum_{ j \in\cup_{\Lin=3}^{\gn}\Ac_{\Lin}}   \Lk_{i}^{[1]} (j)    \geq (n-t -|\Fc|-|\Ac_{2}|) \cdot |\Ac_{2,2}^{[1]} | .  \label{eq:A22condition91875} 
\end{align}

On the other hand, for  any $j\in   \Ac_{ \Lin^{\star}}$, $ \Lin^{\star}\neq   \Lin$ and $ \Lin^{\star},   \Lin \in [1:\gn]$, the term $\sum_{ i \in \Ac_{\Lin}^{[1]}}   \Lk_{i}^{[1]} (j) $ can be upper bounded by  
\begin{align}
\sum_{ i \in \Ac_{\Lin}^{[1]}}   \Lk_{i}^{[1]} (j)  &=    \sum_{ i \in \Ac_{\Lin,  \Lin^{\star}}^{[1]}}   \Lk_{i}^{[1]} (j)   + \sum_{ i \in \Ac_{\Lin}^{[1]}  \setminus \Ac_{\Lin,  \Lin^{\star}}^{[1]}}   \Lk_{i}^{[1]} (j)     \non\\
&=    \sum_{ i \in \Ac_{\Lin,  \Lin^{\star}}^{[1]}}   \Lk_{i}^{[1]} (j)    \label{eq:A22condition91858}  \\
&\leq     |\Ac_{\Lin,  \Lin^{\star}}^{[1]}|    \non\\
&\leq     k-1     \label{eq:A22condition81857}
\end{align}
where \eqref{eq:A22condition91858} results from the fact that $\hv_i^\T  \Meg_{\Lin^{\star}}\neq   \hv_i^\T  \Meg_{\Lin}$ for $i \in \Ac_{\Lin}^{[1]}  \setminus \Ac_{\Lin,  \Lin^{\star}}^{[1]}$ (see  \eqref{eq:Alj11}), which implies that $ \Lk_{i}^{[1]} (j)=0$ for $i \in \Ac_{\Lin}^{[1]}  \setminus \Ac_{\Lin,  \Lin^{\star}}^{[1]}$, $j\in   \Ac_{ \Lin^{\star}}$ and $ \Lin^{\star}\neq   \Lin$; and the last inequality in \eqref{eq:A22condition81857} follows from Lemma~\ref{lm:sizeboundMatrix}. 
With the result in \eqref{eq:A22condition81857}, we can bound that 
 \begin{align}
\sum_{ i \in \Ac_{2,2}^{[1]}}   \Lk_{i}^{[1]} (j)  \leq \sum_{ i \in \Ac_{2}^{[1]}}   \Lk_{i}^{[1]} (j)  \leq  k-1,\quad     \forall j  \in\cup_{\Lin=3}^{\gn}\Ac_{\Lin}    \label{eq:A22condition4747743} 
\end{align}
which also gives the following bound 
 \begin{align}
\sum_{ j \in\cup_{\Lin=3}^{\gn}\Ac_{\Lin}} \sum_{ i \in \Ac_{2,2}^{[1]}}   \Lk_{i}^{[1]} (j)  \leq    (k-1) \cdot \sum_{\Lin=3}^{\gn} |\Ac_{\Lin}|  \label{eq:A22condition43536} 
\end{align}

By combining the results of  \eqref{eq:A22condition91875} and  \eqref{eq:A22condition43536}, the following bound is obvious 
  \begin{align}
 (n-t -|\Fc|-|\Ac_{2}|) \cdot |\Ac_{2,2}^{[1]} |   \leq    (k-1) \cdot \sum_{\Lin=3}^{\gn} |\Ac_{\Lin}|  \label{eq:A22condition92858} 
\end{align}
which also implies that 
 \begin{align}
|\Ac_{2,2}^{[1]}| \leq    \frac{(k-1) \cdot \sum_{\Lin=3}^{\gn} |\Ac_{\Lin}|}{n-t- |\Fc| -|\Ac_{2}| }  ,    \label{eq:A22condition32535}  
\end{align}
where $n-t- |\Fc| -|\Ac_{2}|>0$ holds true under   the conditions of $|\Ac_{2}| \leq    \frac{n-|\Fc|-\sum_{\Lin=3}^{\gn} |\Ac_{\Lin}|}{2}$  and $n\geq 3t+1$.
 
From  the result in \eqref{eq:A22condition32535} we have the following bound  
\begin{align}
|\Ac_{2,2}^{[1]}| &\leq    \frac{(k-1) \cdot \sum_{\Lin=3}^{\gn} |\Ac_{\Lin}|}{n-t- |\Fc| -|\Ac_{2}| }    \non\\
&\leq    \frac{(k-1) \cdot \sum_{\Lin=3}^{\gn} |\Ac_{\Lin}|}{n-t- |\Fc| - \frac{n-|\Fc|-\sum_{\Lin=3}^{\gn} |\Ac_{\Lin}|}{2}}     \label{eq:A22bound2446}  \\
&=    \frac{2(k-1) \cdot \sum_{\Lin=3}^{\gn} |\Ac_{\Lin}|}{n-2t- |\Fc| + \sum_{\Lin=3}^{\gn} |\Ac_{\Lin}| }    \non \\
&=    \frac{2(k-1)  }{\frac{n-2t- |\Fc|}{\sum_{\Lin=3}^{\gn} |\Ac_{\Lin}|} + 1}    \non \\
&\leq     \frac{2(k-1)  }{0 + 1}      \label{eq:A22bound6255}  \\
&=    2(k-1)       \label{eq:A22bound82755}  
\end{align}
where   \eqref{eq:A22bound2446} uses the condition $|\Ac_{2}| \leq    \frac{n-|\Fc|-\sum_{\Lin=3}^{\gn} |\Ac_{\Lin}|}{2}$;  \eqref{eq:A22bound6255} is derived from the identity  that   $\frac{n-2t- |\Fc|}{\sum_{\Lin=3}^{\gn} |\Ac_{\Lin}|} > 0$ in this   case with $\sum_{\Lin=3}^{\gn} |\Ac_{\Lin}| >0$, and given the condition of $n\geq 3t+1$.  

Let us now consider the case with $\Lin = 2$.   In this case we will prove that $|\Ac_{2,2}^{[1]}|=0$ under the condition of $|\Ac_{2}| \leq    \frac{n-|\Fc|-\sum_{\Lin=3}^{\gn} |\Ac_{\Lin}|}{2}$.
Let us assume that $|\Ac_{2,2}^{[1]}| >0$.  Then, by following the steps in \eqref{eq:A22condition66556}-\eqref{eq:A22condition28366}, and given $\Lin = 2$ in this case, we have 
 \begin{align}
0=\sum_{ j \in\cup_{\Lin=3}^{\gn}\Ac_{\Lin}}   \Lk_{i}^{[1]} (j)  \geq n-t -|\Fc|-|\Ac_{2}|,\quad     \forall i\in \Ac_{2,2}^{[1]} .   \label{eq:A22condition28366BB} 
\end{align}
The bound in  \eqref{eq:A22condition28366BB} apparently contradicts the condition of $|\Ac_{2}| \leq    \frac{n-|\Fc|-\sum_{\Lin=3}^{\gn} |\Ac_{\Lin}|}{2} < n-t -|\Fc|$. In other words, the assumption of $|\Ac_{2,2}^{[1]}| >0$ leads to a contradiction. Therefore, it is true that $|\Ac_{2,2}^{[1]}|=0$ under the condition of $|\Ac_{2}| \leq    \frac{n-|\Fc|-\sum_{\Lin=3}^{\gn} |\Ac_{\Lin}|}{2}$, for this case with  $\Lin = 2$.
At this point we complete the proof. 
\end{proof}

\begin{algorithm}
\caption{$\OciorRBC$ protocol with identifier $\ltuple \IDCOOL, \RBCleaderindex \rtuple $.  Code is shown for    $S_{\thisnodeindex}$. }  \label{algm:OciorRBC}
\begin{algorithmic}[1]
\vspace{5pt}    
 
\footnotesize
 \Statex   \emph{//   **   $\OciorRBC$ can be slightly modified to a reliable broadcast protocol without  balancing the communication between the leader the other nodes. In this case,  in the initial phase the leader just broadcasts the whole message to each node. **}

\Statex {\bf \emph{Initial phase}}

\State  Initially set   $k\gets \bigl \lfloor   \frac{ t  }{5 } \bigr\rfloor    +1$; $\OECSI\gets 0; \OECSIFinal\gets 0; \OECsymbolset \gets  \{\}; \OECCorrectSymbolSet\gets \{\}; \Lkset_0\gets \{\}; \Lkset_1\gets \{\}; \Ss_0^{[1]}\gets \{\};\Ss_1^{[1]}\gets \{\};\Ss_0^{[2]}\gets \{\};\Ss_1^{[2]}\gets \{\};   \Me_{i}\gets \defaultvalue;  \Me^{(i)}\gets \defaultvalue;   \ECCEncindicator \gets 0; \SIPhtwo\gets 0;  \Phoneindicator \gets 0; \Phtwoindicator\gets 0; \Phthreeindicator\gets 0$

\State {\bf upon} receiving  a non-empty  message input $\wv$, and if this node is the leader, i.e.,  $\thisnodeindex= \RBCleaderindex$ {\bf do}:   \quad   \emph{//    only for  the leader node  } 
\Indent  
		\State  $[\OECsymbol_1, \OECsymbol_2, \cdots, \OECsymbol_{n}]\gets \ECCEnc(n, k, \wv)$
		\State   $\send$ $\ltuple   \LEADER, \IDCOOL,  \OECsymbol_{j} \rtuple $ to  $S_j$,    $\forall j \in  [n]$       \label{line:ph0Lead} 
 
\EndIndent

\State {\bf upon} receiving   $\ltuple   \LEADER, \IDCOOL,  \OECsymbol_{i} \rtuple $  from  the leader for the first time {\bf do}:  
\Indent  
		\State   $\send$ $\ltuple   \INITIAL, \IDCOOL,   \OECsymbol_{i} \rtuple $ to  all nodes        \label{line:ph0Echo}       \quad\quad \quad \quad\quad\quad \quad   \quad\quad \quad \quad\quad\quad \quad   \quad\quad \quad \quad\quad\quad \quad \quad   \quad\quad \quad   \emph{// echo coded  symbol}   		
\EndIndent

\State {\bf upon} receiving message  $\ltuple   \INITIAL, \IDCOOL,  \OECsymbol_{j} \rtuple $ from $S_j$  for the first time, and $\OECSI=0$   {\bf do}:   \label{line:ph0OECbegin} 	
\Indent  

		\State $\OECsymbolsetInitial \gets \OECsymbolsetInitial\cup  \{j: \OECsymbol_{j}\}$	
 		\If  { $|\OECsymbolsetInitial|\geq  k + t  $}    \label{line:ph0OECCond}   \quad    \quad\quad \quad \quad\quad\quad \quad    \quad\quad \quad \quad\quad\quad \quad   \quad\quad \quad \quad\quad\quad \quad   \quad\quad \quad \quad\quad\quad \quad    \emph{//   online error correcting  (OEC)  }  
			\State   $\tilde{\wv}  \gets \ECCDec(n,  k, \OECsymbolsetInitial)$	
			\State  $ [\OECsymbol_{1}', \OECsymbol_{2}', \cdots, \OECsymbol_{n}'] \gets \ECCEnc (n,  k, \tilde{\wv})$ 
    			\If {at least $k + t$ symbols in $ [\OECsymbol_{1}', \OECsymbol_{2}', \cdots, \OECsymbol_{n}'] $ match with  those in $\OECsymbolsetInitial$, and $\tilde{\wv}$ is non-empty}
 				\State  $\Me_{i} \gets \tilde{\wv}, \Me^{(i)} \gets \tilde{\wv}; \OECSI\gets 1; \Phoneindicator \gets 1$   	  \label{line:RBCph0B} 	
    			\EndIf    
		\EndIf   
	    		
\EndIndent

\Statex {\bf \emph{Phase~1}}       

 \State {\bf upon}  $\Phoneindicator = 1$ {\bf do}:  
\Indent  

 \State    $[y_1^{(i)}, y_2^{(i)}, \cdots, y_{n}^{(i)}]\gets \ECCEnc(n, k, \Me_{i})$     \label{line:RBCECCEnc}
 \State   $\send$  $\ltuple   \SYMBOL, \IDCOOL,  (y_j^{(i)}, y_i^{(i)}) \rtuple$ to $S_j$,  $\forall j \in  [\networksizen]$, and then set $\ECCEncindicator\gets 1$   \  \quad\quad \quad \quad\quad\quad \quad \quad \emph{// exchange coded  symbols}   		 \label{line:RBCECCEncindicator}
  
\EndIndent

\State {\bf upon} receiving   $\ltuple\SYMBOL, \IDCOOL, (y_i^{(j)}, y_j^{(j)}) \rtuple$  from  $S_j$ for the first time  {\bf do}:  
\Indent  
 
	\State  $\wait$ until  $\ECCEncindicator =1$ 
	\If {$ (y_i^{(j)}, y_j^{(j)}) = (y_i^{(i)}, y_j^{(i)})$ }         \label{line:RBCph1MatchCond}
		\State  $\Lkset_1\gets \Lkset_1\cup \{j\}$              \label{line:RBCph1Match}   \quad  \quad \quad \quad\quad \quad \quad \quad \quad \quad\quad \quad \quad  \quad \quad \quad \quad \quad \quad \quad \quad  \quad \quad \quad \quad \emph{//update the set of link indicators}
	\Else 
		\State  $\Lkset_0\gets \Lkset_0\cup \{j\}$            \label{line:RBCph1NotMatch} 
	\EndIf
 
\EndIndent

\State {\bf upon}  $|\Lkset_1|\geq  n- t $,  and  $\ltuple\SIone, \IDCOOL, * \rtuple$  not yet sent {\bf do}:      \label{line:RBCph2OneCond}
\Indent  
           \State	set  $\Ry_i^{[1]} \gets 1$,  $\send$ $\ltuple\SIone, \IDCOOL, \Ry_i^{[1]}\rtuple$ to all nodes, and then set $\Phtwoindicator\gets 1$    \label{line:RBCph2SI1}    \quad   \quad\quad\quad \emph{//set success indicator}
 
\EndIndent

\State {\bf upon}  $|\Lkset_0|\geq   t +1$, and  $\ltuple\SIone, \IDCOOL, * \rtuple$  not yet sent {\bf do}:       \label{line:RBCph2ZeroCond}
\Indent  
           \State	set $\Ry_i^{[1]} \gets 0$,   $\send$ $\ltuple\SIone, \IDCOOL, \Ry_i^{[1]}\rtuple$ to all nodes, 	and  then set $\Phtwoindicator\gets 1$    \label{line:RBCph2B}
\EndIndent

\State {\bf upon} receiving   $\ltuple\SIone, \IDCOOL, \Ry_j^{[1]}\rtuple$  from  $S_j$ for the first time   {\bf do}:    \label{line:RBCph1SS01Cond}  
\Indent  
	\If {$ \Ry_j^{[1]} =1$ }     
		\State  $\wait$ until  $(j \in  \Lkset_1\cup \Lkset_0)\OR(|\Ss_1^{[1]}|\geq  n- t)\OR(|\Ss_0^{[1]}|\geq   t +1)$ 
		\If {$ j \in  \Lkset_1$ }         \label{line:RBCph1SS1Cond} 
			\State  $\Ss_1^{[1]}\gets \Ss_1^{[1]}\cup \{j\}$        \label{line:RBCph1SS1}         \quad  \quad \quad \quad\quad \quad \quad \quad \quad \quad\quad \quad \quad  \quad \quad \quad \quad \quad \quad \quad \quad  \quad \quad \quad \quad \emph{//update the set of success indicator as ones}
		\ElsIf{$ j \in  \Lkset_0$ }
			\State  $\Ss_0^{[1]} \gets \Ss_0^{[1]}\cup \{j\}$       \label{line:RBCph1SS0}     \quad  \quad \quad \quad\quad \quad \quad \quad \quad \quad\quad \quad \quad  \quad \quad \quad \quad \quad \quad \quad \quad  \quad \quad \quad \quad \emph{//mask identified errors  (mismatched links)}
		\EndIf
	\Else 
		\State  $\Ss_0^{[1]} \gets \Ss_0^{[1]}\cup \{j\}$        \label{line:RBCph1SS0NotM}    \quad  \quad \quad \quad\quad \quad \quad \quad \quad \quad\quad \quad \quad  \quad \quad \quad \quad \quad \quad \quad \quad  \quad \quad \quad \quad \emph{//mask identified errors  (mismatched links)}
	\EndIf
 
\EndIndent

\Statex  {\bf \emph{Phase~2}}    

 \State {\bf upon}  $(\Phtwoindicator= 1)\AND(\Ry_i^{[1]} = 0)$, and  $\ltuple\SItwo, \IDCOOL, \Ry_i^{[2]}\rtuple$ not yet sent {\bf do}:  
\Indent  

	\State	set $\Ry_i^{[2]} \gets 0$,   $\send$ $\ltuple\SItwo, \IDCOOL, \Ry_i^{[2]}\rtuple$ to all nodes  \label{line:RBCph2CSI0}    \quad\quad\quad\quad \quad \quad\quad\quad\quad \quad \quad\quad\quad\quad \quad   \emph{//update success indicator }
	
\EndIndent

 \State {\bf upon}  $(\Phtwoindicator= 1)\AND(\Ry_i^{[1]} = 1) \AND(|\Ss_1^{[1]}|\geq  n- t)$, and  $\ltuple\SItwo, \IDCOOL, \Ry_i^{[2]}\rtuple$ not yet sent   {\bf do}:   \label{line:RBCph2ACond} 
\Indent   

           \State	set  $\Ry_i^{[2]} \gets 1$, $\SIPhtwo\gets 1$, and   $\send$ $\ltuple \SItwo, \IDCOOL,\Ry_i^{[2]}\rtuple$ to all nodes   \label{line:RBCph2ASI1} 
	
\EndIndent

 \State {\bf upon}  $|\Ss_0^{[1]}|\geq   t +1$, and  $\ltuple\SItwo, \IDCOOL, \Ry_i^{[2]}\rtuple$ not yet sent   {\bf do}:  
\Indent  

           \State	set $\Ry_i^{[2]} \gets 0$,   $\send$ $\ltuple \SItwo, \IDCOOL, \Ry_i^{[2]}\rtuple$ to all nodes    \label{line:RBCph2BSI0}     
	
\EndIndent

 \State {\bf upon} receiving   $\ltuple\SItwo, \IDCOOL, \Ry_j^{[2]}\rtuple$  from  $S_j$ for the first time  {\bf do}:     \label{line:RBCph2SS01Cond}     
\Indent  
	\If {$ \Ry_j^{[2]} =1$ }     
		\State  $\wait$ until $(j \in  \Lkset_1\cup \Lkset_0)\OR(|\Ss_1^{[2]}|\geq  n- t)\OR(|\Ss_0^{[2]}|\geq   t +1)$ 
		\If {$ j \in  \Lkset_1$ }     
			\State  $\Ss_1^{[2]}\gets \Ss_1^{[2]}\cup \{j\}$          
		\ElsIf{$ j \in  \Lkset_0$ }
			\State  $\Ss_0^{[2]} \gets \Ss_0^{[2]}\cup \{j\}$            
		\EndIf
	\Else 
		\State  $\Ss_0^{[2]} \gets \Ss_0^{[2]}\cup \{j\}$            \label{line:RBCph2SS0NM}     
	\EndIf
 
\EndIndent

\algstore{COOLDRBC}

\end{algorithmic}
\end{algorithm}

\begin{algorithm}
\begin{algorithmic}[1]
\algrestore{COOLDRBC}
\vspace{5pt}    
 \footnotesize

\State {\bf upon}   $|\Ss_{\Vr}^{[2]}|\geq  n- t$,     for a  $\Vr \in \{1,0\}$,  and $\ltuple\READY, \IDCOOL, * \rtuple$  not yet sent {\bf do}:      \label{line:BRCReadyCondition}   
\Indent  
		\State $\send$ $\ltuple\READY, \IDCOOL, \Vr \rtuple$ to  all nodes  	  \label{line:BRCReadySendA}
\EndIndent

\State {\bf upon} receiving   $t+1$  $\ltuple \READY, \IDCOOL, \Vr  \rtuple$ messages  from different   nodes for the same $\Vr$ and $\ltuple\READY, \IDCOOL, * \rtuple$  not yet sent {\bf do}:        \label{line:BRCRelialbeBegin}   
\Indent  
		\State $\send$ $\ltuple\READY, \IDCOOL, \Vr \rtuple$ to  all nodes	   \label{line:BRCRelialbeBeginBB}   
\EndIndent

\State {\bf upon} receiving   $2t+1$  $\ltuple \READY, \IDCOOL, \Vr  \rtuple$ messages  from different nodes   for the same $\Vr$ {\bf do}:    \label{line:BRCRelialbeEnd}  
\Indent  
	\If {$\ltuple\READY, \IDCOOL, * \rtuple$  not yet sent }     
		\State $\send$ $\ltuple\READY, \IDCOOL, \Vr \rtuple$ to  all nodes    \label{line:BRCReadySendC}
	\EndIf
	\State  set $\VrOutput\gets \Vr$             \label{line:BRCVrOutput}     
	\If {$\VrOutput =0$ }     
		\State set $\Me^{(i)}\gets \defaultvalue$, then $\Output$   $\Me^{(i)}$ and $\terminate$  	 \label{line:BRCoutputdefault}     	 \quad\quad\quad\quad   \quad\quad\quad\quad \quad\quad     \emph{//  $\defaultvalue$ is a default value  }    
	\Else
		\State  set $\Phthreeindicator\gets 1$        \label{line:BRCph3}

	\EndIf
\EndIndent

\Statex

\Statex  {\bf \emph{Phase~3}}

\State {\bf upon} $\Phthreeindicator= 1$ {\bf do}:     \quad \quad\quad\quad\quad\quad  \quad \quad\quad \quad \quad\quad\quad\quad\quad  \quad \quad\quad\quad\quad\quad\quad \quad      \emph{//     only after executing Line~\ref{line:BRCph3}}
\Indent  
\If { $\SIPhtwo = 1$}          \label{line:BRCph3SIPhtwo}
	\State $\Output$   $\Me^{(i)}$ and $\terminate$     \label{line:BRCph3SIPhtwoOutput}
		 
\Else
 
			\State  $\wait$ until receiving $t +1$ $\ltuple\SYMBOL, \IDCOOL, (y_i^{(j)},  *) \rtuple$  messages, $\forall j \in   \Ss_1^{[2]}$, for the same    $y_i^{(j)} =y^{\star}$,  for some   $y^{\star}$   \label{line:BRCph3MajorityRuleCond}
			\State  $y_i^{(i)} \gets y^{\star}$   \label{line:BRCph3MajorityRule}   \quad\quad\quad\quad \quad\quad\quad\quad\quad\quad\quad\quad\quad\quad\quad\quad   \quad\quad\quad\quad\quad\quad\quad\quad \quad\quad \emph{// update coded symbol based on  majority rule}  
	\State   $\send$   $\ltuple \CORRECTSYMBOL, \IDCOOL, y_i^{(i)}  \rtuple$  to  all nodes      \label{line:BRCph3SendCorrectSymbols} 
			\State  $\wait$ until  $\OECSIFinal=1$ 	
			\State $\Output$   $\Me^{(i)}$ and $\terminate$     \label{line:BRCph3Output2}	
	
\EndIf

\EndIndent

\State {\bf upon} receiving   $\ltuple\CORRECTSYMBOL, \IDCOOL, y_j^{(j)} \rtuple$  from  $S_j$ for the first time,   $j\notin \OECCorrectSymbolSet$, and $\OECSIFinal=0$   {\bf do}:  \label{line:Ph3OECCond}
\Indent  
	\State $\OECCorrectSymbolSet[j] \gets y_j^{(j)}$    
 	\If  { $|\OECCorrectSymbolSet|\geq  k + t  $}     \label{line:OECbegin}    \quad    \quad\quad \quad \quad\quad\quad \quad    \quad\quad \quad \quad\quad\quad \quad   \quad\quad \quad \quad\quad\quad \quad   \quad\quad \quad \quad\quad\quad \quad    \emph{//   online error correcting  (OEC)  }  
			\State   $\MVBAOutputMsg  \gets \ECCDec(n,  k , \OECCorrectSymbolSet)$	     
			\State  $[y_{1}, y_{2}, \cdots, y_{n}] \gets \ECCEnc (n,  k, \MVBAOutputMsg)$ 
    			\If {at least $k + t$ symbols in $[y_{1}, y_{2}, \cdots, y_{n}]$ match with  those in $\OECCorrectSymbolSet$}
 				\State  $ \Me^{(i)} \gets \MVBAOutputMsg; \OECSIFinal\gets 1$     \label{line:OECend}	 
    			\EndIf    

	\EndIf   
		     
\EndIndent

\State {\bf upon} having received   both $\ltuple\SYMBOL, \IDCOOL, (y_i^{(j)}, y_j^{(j)})\rtuple$ and   $\ltuple \SItwo, \IDCOOL, 1\rtuple$ messages from  $S_j$, and $j\notin \OECCorrectSymbolSet$, and  $\OECSIFinal=0$   {\bf do}:    \label{line:RBCPh3OECSecCond}
\Indent  
	\State $\OECCorrectSymbolSet[j] \gets y_j^{(j)}$     
	\State run the OEC steps as in Lines~\ref{line:OECbegin}-\ref{line:OECend}      \label{line:Ph3OECAllEnd}
\EndIndent

\end{algorithmic}
\end{algorithm}

\newpage

\section{$\OciorRBC$}    \label{sec:OciorRBC}

This proposed $\OciorRBC$ is an asynchronous error-free Byzantine reliable broadcast  protocol.  
$\OciorRBC$ doesn't rely on any cryptographic assumptions such as signatures or hashing. 
This proposed $\OciorRBC$ is an extension of $\OciorCOOL$.

\subsection{Overview of $\OciorRBC$}    \label{sec:OverviewOciorRBC}

The proposed  $\OciorRBC$ is described in Algorithm~\ref{algm:OciorRBC}. In the following, we provide an overview of the proposed protocol.

\subsubsection{Initial phase} $\OciorRBC$ is a \emph{balanced} reliable broadcast protocol where communication overhead is distributed evenly between the leader and the other nodes. In this initial phase, the goal is to multicast the leader's message to the distributed nodes with balanced communication. This initial phase guarantees that, if the leader is honest  then every honest node eventually outputs the message sent from the leader. 

Without considering balanced communication, the leader could simply send the entire initial message to each node during the initial phase. However, this simple multicasting would result in a heavy communication load for the leader. To reduce this load, the leader sends different coded symbols to different nodes in the initial phase, with these symbols being encoded from the initial message (Line~\ref{line:ph0Lead}). Each node then  echoes the received symbol to all other nodes  (Line~\ref{line:ph0Echo}). Upon receiving the coded symbols, each node conducts online error correction to decode the message sent by the leader (Lines~\ref{line:ph0OECbegin}-\ref{line:RBCph0B}).

\subsubsection{Phase~1}  The goal of Phase~1 is  to exchange coded information symbols and mask  inconsistent messages.    
This phase, together with Phase~2, guarantee that all honest nodes who set their success indicators to one in Phase~2 should have the same input message at the beginning of Phase~$1$.

In this phase, Node~$i$ encodes the message $\Me_{i}$ delivered from the initial phase into coded symbols  $[y_1^{(i)}, y_2^{(i)}, \cdots, y_{n}^{(i)}]$ by using error correction code, for $i\in [\networksizen]$ (Line~\ref{line:RBCECCEnc}).  Then, Node~$i$ sends  $\ltuple   \SYMBOL, \IDCOOL,  (y_j^{(i)}, y_i^{(i)}) \rtuple$ to Node~$j$,  $\forall j \in  [\networksizen]$  (Line~\ref{line:RBCECCEncindicator}).  Upon receiving   $\ltuple\SYMBOL, \IDCOOL, (y_i^{(j)}, y_j^{(j)}) \rtuple$  from  Node~$j$ for the first time, Node~$i$ checks if  the received observation $(y_i^{(j)}, y_j^{(j)})$ matches   its available  local observation $(y_i^{(i)}, y_j^{(i)})$  (Line~\ref{line:RBCph1MatchCond}).  Node~$i$  includes the index~$j$ into the set $\Lkset_1$ if $(y_i^{(j)}, y_j^{(j)}) = (y_i^{(i)}, y_j^{(i)})$, else puts the index~$j$ into the set $\Lkset_0$ (Lines~\ref{line:RBCph1Match} and \ref{line:RBCph1NotMatch}). 
  If  $(y_i^{(j)}, y_j^{(j)}) = (y_i^{(i)}, y_j^{(i)})$, it can be considered that the link indicator between Node~$i$ and Node~$j$, denoted by  $\Lk_{i} (j)$, is  $\Lk_{i} (j)=1$. On the other hand,    $(y_i^{(j)}, y_j^{(j)}) \neq  (y_i^{(i)}, y_j^{(i)})$ implies that $\Lk_{i} (j)=0$. 
  
 When  $|\Lkset_1|\geq  n- t $, Node~$i$ sets the success indicator at Phase~1 as $\Ry_i^{[1]} = 1$  and sends  $\ltuple\SIone, \IDCOOL, \Ry_i^{[1]}\rtuple$ to all nodes  (Line~\ref{line:RBCph2SI1}).
On the other hand, when $|\Lkset_0|\geq   t +1$, Node~$i$   sets $\Ry_i^{[1]} = 0$,   and sends $\ltuple\SIone, \IDCOOL, \Ry_i^{[1]}\rtuple$ to all nodes (Line~\ref{line:RBCph2B}). 
It is worth noting  that the two conditions of $|\Lkset_1|\geq  n- t $ and $|\Lkset_0|\geq   t +1$ cannot be satisfied at the same time. 

Upon receiving   $\ltuple\SIone, \IDCOOL, 1 \rtuple$  from  Node~$j$,  Node~$i$ puts  the index~$j$ into the set $\Ss_1^{[1]}$ once Node~$i$ has received matched observation $(y_i^{(j)}, y_j^{(j)})$ from  Node~$j$, i.e.,  $j\in \Lkset_1$ (Line~\ref{line:RBCph1SS1}).  If Node~$i$ has received unmatched observation from  Node~$j$, i.e., $j\in \Lkset_0$, then   Node~$i$ puts  the index~$j$ into the set $\Ss_0^{[1]}$ (Line~\ref{line:RBCph1SS0}).  
Upon receiving   $\ltuple\SIone, \IDCOOL, 0 \rtuple$  from  Node~$j$,  Node~$i$ directly puts  the index~$j$ into the set $\Ss_0^{[1]}$  (Line~\ref{line:RBCph1SS0NotM}). 

\subsubsection{Phase~2}  One goal of Phase~2 is to mask the remaining inconsistent messages so that all honest nodes who set their success indicators to one in this phase should have the same input message at the beginning of Phase~$1$. Another goal of Phase~2 is to reach a consensus on whether to proceed to the next phase or terminate at this phase, together.

If the success indicator was set as $\Ry_i^{[1]} = 0$ at Phase~1, or  if $|\Ss_0^{[1]}|\geq   t +1$, then  Node~$i$ sets the success indicator at Phase~2 as $\Ry_i^{[2]} = 0$ (Lines~\ref{line:RBCph2CSI0} and \ref{line:RBCph2BSI0}). If the success indicator was set as $\Ry_i^{[1]} = 1$ at Phase~1 and given $|\Ss_1^{[1]}|\geq  n- t$, then  Node~$i$ sets the success indicator at Phase~2 as $\Ry_i^{[2]} = 1$ and sets a ready indicator as $\SIPhtwo= 1$ (Line~\ref{line:RBCph2ASI1}).  After setting the value of $\Ry_i^{[2]}$,  Node~$i$ sends  $\ltuple \SItwo, \IDCOOL,\Ry_i^{[2]}\rtuple$ to all nodes. 

Upon receiving   $\ltuple\SItwo, \IDCOOL, \Ry_j^{[2]}\rtuple$  from  $S_j$,  Node~$i$ conducts a process to decide whether to include the index $j$ in $\Ss_1^{[2]}$ or $\Ss_0^{[2]}$ (Lines~\ref{line:RBCph2SS01Cond}-\ref{line:RBCph2SS0NM}), similarly to the process   in Phase~1 upon receiving $\ltuple\SIone, \IDCOOL, \Ry_j^{[1]} \rtuple$ (Lines~\ref{line:RBCph1SS01Cond}-\ref{line:RBCph1SS0NotM}) 

In this phase, the distributed honest nodes also conduct a process to reach a consensus on whether to proceed to the next phase or terminate at this phase, together  (Lines~\ref{line:BRCReadyCondition}-\ref{line:BRCph3}).

\subsubsection{Phase~3}  

 Phase~3 is initiated only after the distributed nodes have decided to proceed to this phase, together (Line~\ref{line:BRCph3}).  The goal of Phase~3 is to calibrate the coded symbols based on the majority rule to ensure consistent consensus outputs from honest nodes. 
 In this phase, if Node~$i$ has set the success indicator at Phase~2 as $\Ry_i^{[2]} = 1$  (or has set $\SIPhtwo= 1$), then Node~$i$ outputs  the message  updated in the initial phase (Line~\ref{line:RBCph0B}) and then terminates (Line~\ref{line:BRCph3SIPhtwoOutput}).  

In this phase, if Node~$i$ hasn't set the success indicator at Phase~2 as $\Ry_i^{[2]} = 1$ yet, then Node~$i$  waits until receiving $t +1$ $\ltuple\SYMBOL, \IDCOOL, (y_i^{(j)},  *) \rtuple$  messages, $\forall j \in   \Ss_1^{[2]}$, for the same    $y_i^{(j)} =y^{\star}$,  for some   $y^{\star}$  (Line~\ref{line:BRCph3MajorityRuleCond}) and then updated its coded symbol $y_i^{(i)}$ as  $y_i^{(i)} = y^{\star}$  based on the majority rule (Line~\ref{line:BRCph3MajorityRule}). Then  Node~$i$ sends the message  $\ltuple \CORRECTSYMBOL, \IDCOOL, y_i^{(i)}  \rtuple$ with   updated symbol  to  all nodes.     
 
In this phase,   Node~$i$   conducts the online error correction to decode the message  (Lines~\ref{line:Ph3OECCond}-\ref{line:Ph3OECAllEnd}), based on the received updated symbols  (Lines~\ref{line:Ph3OECCond}-\ref{line:OECend}) and symbols from nodes that have set their  success indicators to ones at Phase~$2$ (Lines~\ref{line:RBCPh3OECSecCond}-\ref{line:Ph3OECAllEnd}). After the completion of online error correction,  Node~$i$ outputs the decoded message and terminates (Line~\ref{line:BRCph3Output2}).

\subsection{Analysis of $\OciorRBC$}    \label{sec:AnalysisOciorRBC}

The analysis of $\OciorRBC$  follows closely that of $\OciorCOOL$ shown in Section~\ref{sec:AnalysisOciorCOOL}. 
In the analysis here we will use similar notations previously used for $\OciorCOOL$. 
Similarly to the analysis for $\OciorCOOL$,  without loss of generality we just focus on the case with  $| \Fc |   = t$, where $\Fc$  is defined as  the set of       dishonest nodes.  
Here we use $\Me_{i}^{[0]}$ to denote the value of $\Me_{i}$ updated at Phase~$0$. 
If Node~$i$ never updates  the value of $\Me_{i}$ before termination, then  $\Me_{i}^{[0]}$ is considered to be a default value $\Me_{i}^{[0]}=  \defaultvalue$. 
We define some groups of honest nodes as
  \begin{align}
 \Ac_{\Lin} \defeq &  \{  i:    \Me_i^{[0]} =  \Meg_{\Lin},  \  i \notin  \Fc  , \ i \in [1:n]\}, \quad   \Lin \in [1 : \gn]     \label{eq:RBCAell00}   \\
\Ac_{\Lin}^{[p]} \defeq&   \{  i:  \Ry_i^{[p]} =1, \Me_i^{[0]} =  \Meg_{\Lin},  \  i \notin  \Fc, \ i \in [1:n]\}, \quad   \Lin \in [1 : \gn^{[p]}], \quad  p\in \{1,2\}      \label{eq:RBCAell}   \\
\Bc^{[p]} \defeq  & \{  i:  \Ry_i^{[p]} =0  ~\text{or $\Ry_i^{[p]}$ has never been set},  \  i \notin  \Fc, \ i \in [1:n] \}, \quad  p\in \{1,2\}    \label{eq:RBCBdef01} 
 \end{align} 
 for  some different  $\ell$-bit  values $\Meg_{1}, \Meg_{2}, \cdots, \Meg_{\gn-1}$, where only one of them could be a default value $\defaultvalue$; and for some non-negative integers $\gn, \gn^{[1]},\gn^{[2]}$  such that $ \gn^{[2]} \leq \gn^{[1]} \leq \gn$.    We use the same notions of $ \Ac_{\Lin,j},  \Ac_{\Lin,\Lin},  \Ac_{\Lin,j}^{[p]},  \Ac_{\Lin,\Lin}^{[p]}$  defined in \eqref{eq:Alj}-\eqref{eq:All11}. 
 We use   $\Lk_{i} (j), \Lk_{i}^{[1]} (j), \Lk_{i}^{[2]} (j)\in \{0,1\}$ to denote the   link indicator between Node $i$ and Node~$j$,   its value  at   Phase~1,  and its value  at Phase~2, respectively,  from the view of Node~$i$, defined by  
  \begin{numcases}  
{\Lk_{i}^{[1]} (j)= \Lk_i (j)=} 
     1   &    if   \    $(y_i^{(j)}, y_j^{(j)}) = (y_i^{(i)}, y_j^{(i)})$           			 \label{eq:lkindicator}  \\
  0  &  else            		\non  
\end{numcases}
and
  \begin{numcases}  
{\Lk_{i}^{[2]} (j)=} 
     \Lk_{i}^{[1]} (j)    &    if   \       $\Ry_i^{[1]}=\Ry_j^{[1]}=1$             			 \label{eq:lkindicator2}  \\
  0  &  else         .   		\non  
\end{numcases}
 It is worth mentioning that if  $(y_{i}^{(i)}, y_j^{(i)})$ are never sent by Node~$i$, or  $(y_{i}^{(j)}, y_j^{(j)})$  are never received at Node~$i$, then $\Lk_{i}^{[1]} (j) =\Lk_{i} (j) = 0$.  
 Similarly, if  $\Ry_i^{[1]}$ is never sent by Node~$i$, or  $\Ry_j^{[1]}$  is never received at Node~$i$, then $\Lk_{i}^{[2]} (j) = 0$.  
 In our setting, for any honest Node $i$ and Node $j$,   eventually they will have the same view on $ (y_i^{(j)}, y_j^{(j)}),  (y_i^{(i)}, y_j^{(i)}), \Ry_i^{[1]}, \Ry_j^{[1]}$. Therefore, it holds true that eventually $\Lk_i (j)=\Lk_j(i)$, $\Lk_{i}^{[1]} (j)= \Lk_{j}^{[1]} (i)$, and $\Lk_{i}^{[2]} (j)= \Lk_{j}^{[2]} (i)$, for any $i,j \notin \Fc$.  In the analysis here we focus on the final values of the link indicators.

The main results of  $\OciorRBC$ are summarized in Theorems~\ref{thm:BRCtotality}-\ref{thm:BRCPerformance}. Theorems~\ref{thm:BRCtotality}-\ref{thm:BRCvalidity} reveal that, given $n\geq 3t+1$, the  totality, validity and consistency  conditions are all satisfied in all executions (\emph{error-free}).  
Theorem~\ref{thm:BRCPerformance} shows that $\OciorRBC$ is optimal in terms of  communication complexity, round complexity and resilience.

 \begin{theorem}  [Totality and Consistency]  \label{thm:BRCtotality}
In $\OciorRBC$,  given $n\geq 3t+1$, if one  honest node outputs a value $\Me^{\star}$, then every honest node  eventually outputs a value $\Me^{\star}$, for some $\Me^{\star}$. 
\end{theorem}
\begin{proof}
Lemma~\ref{lm:RBCvoutput} reveals that,    if one  honest node sets the value of $\VrOutput$  in  Line~\ref{line:BRCVrOutput} as  $\VrOutput= \Vr^\star$  for  a binary value $\Vr^\star\in \{1,0\}$, then  every honest node  eventually sets  $\VrOutput= \Vr^\star$. 
Based on this result, if  one  honest node sets  $\VrOutput=0$ then  every honest node  eventually sets  $\VrOutput=0$. In this case,  every honest node  eventually  outputs a default value   $\defaultvalue$ (see Line~\ref{line:BRCoutputdefault}).

Based on the result of Lemma~\ref{lm:RBCvoutput}, if  one  honest node sets  $\VrOutput=1$   (see Line~\ref{line:BRCph3}), then   every honest node  eventually sets  $\VrOutput=1$.  What remains to be proved is that  in this case  all honest   nodes will eventually output the same value  at Phase~3. 

If an honest node sets  $\VrOutput=1$  and  $\SIPhtwo = 1$, then this node outputs the value of   $\Me^{(i)}$ (see Line~\ref{line:BRCph3SIPhtwoOutput}). 
 Lemma~\ref{lm:BRCuniquegroup} reveals that all of the  honest nodes who set $\SIPhtwo = 1$ at Phase~$2$ should have  the same input message $\Me^{(i)}=\Me^{\star}$ at the beginning of Phase~$1$, for some $\Me^{\star}$. 
 Thus, all of the honest node who set  $\VrOutput=1$ and  $\SIPhtwo = 1$ eventually output the same value     $\Me^{\star}$. 

If an honest node sets  $\VrOutput=1$ and  $\SIPhtwo = 0$, it can be shown that this node  will eventually  output the same value   $\Me^{\star}$ in Line~\ref{line:BRCph3Output2}. 
If an honest node sets  $\VrOutput=1$, it means that this node has  received at least   $2t+1$  $\ltuple \READY, \IDCOOL,  1  \rtuple$ messages  from different nodes, which also implies that at least one honest node has received  $n-t$  $\ltuple \SItwo, \IDCOOL, 1  \rtuple$  messages from different nodes at Phase~$2$ (see Line~\ref{line:BRCReadyCondition}). 
In other words,  if an honest node sets  $\VrOutput=1$, then at least $n-2t$ honest nodes have sent out the message  $\ltuple \SItwo, \IDCOOL, 1  \rtuple$  at Phase~$2$.  It is worth noting that if an  honest Node~$i$  sends out a message  $\ltuple \SItwo, \IDCOOL, 1  \rtuple$  at Phase~$2$, this node should have sent $\ltuple   \SYMBOL, \IDCOOL,  (y_j^{(i)}, y_i^{(i)}) \rtuple$ to $S_j$,  $\forall j \in  [\networksizen]$ at Phase~$1$. 
On the other hand, Lemma~\ref{lm:BRCuniquegroup} reveals that the honest nodes who send out  $\ltuple\SItwo, \IDCOOL, 1\rtuple$ in Phase~$2$ should have  the same input message $\Me^{\star}$ at the beginning of Phase~$1$, for some $\Me^{\star}$.
Thus, if an honest Node~$i$ sets  $\VrOutput=1$, and   $\SIPhtwo = 0$, then it will eventually receives at least $n-2t\geq t +1$ matching $\ltuple\SYMBOL, \IDCOOL, (y_i^{(j)},  *) \rtuple$  messages from the honest nodes within $\Ss_1^{[2]}$, for one and only one value $y_i^{(j)} = \ECCEnc_i (n, k, \Me^{\star})$, where $\ECCEnc_i (n, k, \Me^{\star})$ denotes the $i$th  symbol encoded from message $\Me^{\star}$.
In this case, Node~$i$ will set $y_i^{(i)} \gets \ECCEnc_i (n, k, \Me^{\star})$   in Line~\ref{line:BRCph3MajorityRule}, and send  $\ltuple \CORRECTSYMBOL, \IDCOOL, y_i^{(i)}  \rtuple$  to  all nodes   in Line~\ref{line:BRCph3SendCorrectSymbols}.  
At this point,  every symbol $y_j^{(j)}$ collected in $\OECCorrectSymbolSet$ for $j\notin \Fc$ should be the  symbol  encoded from the same message  $\Me^{\star}$, where $\Fc$ denotes the set of dishonest nodes. Therefore,   every honest node who sets  $\VrOutput=1$ and   $\SIPhtwo = 0$ will eventually decode the message $\Me^{\star}$ with OEC decoding and output $\Me^{\star}$ in Line~\ref{line:BRCph3Output2}. 
 \end{proof}

 \begin{theorem}  [Validity]  \label{thm:BRCvalidity}
Given $n\geq 3t+1$,    if the leader is  honest and inputs a value $\wv$, then every honest node eventually outputs $\wv$ in $\OciorRBC$.     
\end{theorem}
\begin{proof}
If the leader is  honest and inputs a value $\wv$, then each symbol  $\OECsymbol_{j}$  in $\ltuple   \LEADER, \IDCOOL,  \OECsymbol_{j} \rtuple $ sent from the leader (see  Line~\ref{line:ph0Lead}) or in $\ltuple   \INITIAL, \IDCOOL,   \OECsymbol_{j} \rtuple $ echoed by the honest node (see Line~\ref{line:ph0Echo}) should be encoded from $\wv$. Thus, at Phase~$0$ every honest node will eventually decode the same message $\wv$ with OEC decoding (see Lines~\ref{line:ph0OECCond}-\ref{line:RBCph0B}), if this node hasn't output a value yet.

Based on the above conclusion,  if an  honest node starts Phase~$1$,  it should have already set the value  of   $\Me_{i}$ and $\Me^{(i)}$ as $\Me_{i} = \Me^{(i)}=\wv$ (see Line~\ref{line:RBCph0B}). Therefore, all symbols $(y_i^{(j)}, y_j^{(j)})$ sent in the messages $\ltuple\SYMBOL, \IDCOOL, (y_i^{(j)}, y_j^{(j)}) \rtuple$  by any honest nodes should be encoded from the message $\wv$, which implies that the condition  of  $(y_i^{(j)}, y_j^{(j)}) = (y_i^{(i)}, y_j^{(i)})$ should be satisfied for any $i, j \notin \Fc$ (see Line~\ref{line:RBCph1MatchCond}). 
This suggests that the condition of  $|\Lkset_0|\geq   t +1$ should not be satisfied at any honest node (see Line~\ref{line:RBCph2ZeroCond}) and that no honest node will set $\Ry_i^{[1]} = 0$ or send out  $\ltuple\SIone, \IDCOOL, 0 \rtuple$  at Phase~$1$ (see Line~\ref{line:RBCph2B}).  
Similarly,  no honest node will set $\Ry_i^{[2]} = 0$ or send out  $\ltuple \SItwo, \IDCOOL, 0 \rtuple$  at Phase~$2$.
In this case,  at least one honest node eventually  receives at least $2t+1$  $\ltuple \READY, \IDCOOL, 1  \rtuple$ messages and sets $\VrOutput\gets 1$, and then outputs $\Me^{(i)}=\wv$  at Line~\ref{line:BRCph3SIPhtwoOutput}. 
From Theorem~\ref{thm:BRCtotality}, if one  honest node outputs a value $\wv$, then every honest node  eventually outputs a value $\wv$.  
\end{proof}

\begin{lemma}    \label{lm:RBCvoutput}
In  $\OciorRBC$,  if one  honest node sets the value of $\VrOutput$  in  Line~\ref{line:BRCVrOutput} as  $\VrOutput= \Vr^\star$  for  a binary value $\Vr^\star\in \{1,0\}$, then  every honest node  eventually sets  $\VrOutput= \Vr^\star$. 
\end{lemma}
\begin{proof}
Let us consider the case that one  honest node sets the value of $\VrOutput$  in  Line~\ref{line:BRCVrOutput} as  $\VrOutput= \Vr^\star$  for  a binary value $\Vr^\star\in \{1,0\}$. 
In this case,  the  node   setting $\VrOutput= \Vr^\star$ should have received at least $2t+1$  $\ltuple \READY, \IDCOOL, \Vr^\star  \rtuple$ messages (see Line~\ref{line:BRCRelialbeEnd}). 
It means that at least $t+1$  $\ltuple \READY, \IDCOOL, \Vr^\star  \rtuple$ messages have been sent out from honest nodes. 
On the other hand,  if two honest nodes   send out   messages $\ltuple\READY, \IDCOOL, \Vr^\star \rtuple$ and $\ltuple\READY, \IDCOOL, \Vr' \rtuple$, respectively, then $\Vr^\star=\Vr'$ (see Lemma~\ref{lm:RBCsmaeready}).  
Therefore, in this case, each honest node eventually sends out  a message $\ltuple\READY, \IDCOOL, \Vr^\star \rtuple$ (see Lines~\ref{line:BRCRelialbeBegin}-\ref{line:BRCRelialbeBeginBB}). 
Thus, each honest node eventually receives at least $2t+1$  $\ltuple \READY, \IDCOOL, \Vr^\star  \rtuple$ messages, which suggests that each honest node eventually goes to  Line~\ref{line:BRCRelialbeEnd} and then set $\VrOutput= \Vr^\star$ in  Line~\ref{line:BRCVrOutput}. 
\end{proof}

\begin{lemma}    \label{lm:RBCsmaeready}
In  $\OciorRBC$, given $n\geq 3t+1$,   if two honest nodes   send out   messages $\ltuple\READY, \IDCOOL, \Vr^\star \rtuple$ and $\ltuple\READY, \IDCOOL, \Vr' \rtuple$, respectively, then $\Vr^\star=\Vr'$.  
\end{lemma}
\begin{proof}
If one honest node   sends out  a  message $\ltuple\READY, \IDCOOL, \Vr^\star \rtuple$ (see Lines~\ref{line:BRCReadySendA}, \ref{line:BRCRelialbeBeginBB} and \ref{line:BRCReadySendC}), it means that at least one honest node has received at least $n-t$ $\ltuple \SItwo, \IDCOOL, \Vr^\star   \rtuple$  messages from different  nodes   (see  Line~\ref{line:BRCReadyCondition}), for  a binary value $\Vr^\star\in \{1,0\}$. 
In this case, at least  $n-2t$ honest nodes have sent out the same message $\ltuple \SItwo, \IDCOOL, \Vr^\star   \rtuple$.  
 
Similarly, if one honest node   sends out  a  message $\ltuple\READY, \IDCOOL, \Vr' \rtuple$, it means that at least one honest node has received at least $n-t$   $\ltuple \SItwo, \IDCOOL, \Vr'   \rtuple$  messages from different  nodes, for  a binary value $\Vr'\in \{1,0\}$.

In   $\OciorRBC$,  each honest node sends out at most one  message $\ltuple\SItwo, \IDCOOL, * \rtuple$.  Thus, if  $n-2t$ honest nodes have sent out the same message $\ltuple \SItwo, \IDCOOL, \Vr^\star   \rtuple$,  it is impossible to have    $n-t$  nodes sending out a  different message $\ltuple \SItwo, \IDCOOL, \Vr'   \rtuple$, for  $\Vr'\neq \Vr^\star$, because $n-(n-2t)<n-t$ given $n\geq 3t+1$.  Therefore, if two honest nodes   send out   messages $\ltuple\READY, \IDCOOL, \Vr^\star \rtuple$ and $\ltuple\READY, \IDCOOL, \Vr' \rtuple$, respectively, then $\Vr^\star=\Vr'$. 
\end{proof}

\begin{lemma}    \label{lm:BRCuniquegroup}
In  $\OciorRBC$, given $n\geq 3t+1$, all of the honest nodes who set $\SIPhtwo = 1$ or send out  $\ltuple\SItwo, \IDCOOL, 1\rtuple$ in Phase~$2$  should have  the same input message $\Me^{\star}$ at the beginning of Phase~$1$, for some $\Me^{\star}$, i.e.,  it holds true that $\gn^{[2]} \leq 1$.
\end{lemma}
\begin{proof}
The proof of this lemma is similar to that of Lemma~\ref{lm:uniquegroup}. 
At first, from Lemma~\ref{lm:BRCeta2bound2} it is concluded that   $\gn^{[2]} \leq 2$.
Next, we   argue that the case of $\gn^{[2]} = 2$ does not exist.   Let us assume that $\gn^{[2]} = 2$. 
Under the assumption of  $\gn^{[2]} = 2$,  it holds true that $\gn^{[1]} =  2$ (see  Lemma~\ref{lm:BRCeta22212}). However, if $\gn^{[1]} = 2$  then it implies that $\gn^{[2]} \leq   1$ (see  Lemma~\ref{lm:BRCeta1231}), which contradicts the assumption of  $\gn^{[2]} = 2$.
Therefore, the case of $\gn^{[2]} = 2$ does not exist, which, together with the result $\gn^{[2]} \leq 2$, concludes that $\gn^{[2]} \leq 1$.
\end{proof}

\begin{theorem}  [Communication, Round, and Resilience]  \label{thm:BRCPerformance}
For the consensus on an $\Lh$-bit message, and given $n\geq  3t+1$, the total communication complexity of $\OciorRBC$ is  $O(\max\{n\Lh, n^2 \log n \})$  bits, while the communication per node is  $O(\max\{\Lh, n \log n \})$  bits. Additionally, the round  complexity of $\OciorRBC$ is  $7$  asynchronous rounds. Without considering balance communication, the  round  complexity of $\OciorRBC$ is  $6$  rounds.
\end{theorem}
\begin{proof}
The proposed $\OciorRBC$  satisfies the totality, validity and consistency  conditions  in all executions,  given   $n\geq  3t+1$ (see Theorems~\ref{thm:BRCtotality}-\ref{thm:BRCvalidity}). 
For the proposed $\OciorRBC$, the  communication is involved in Lines~\ref{line:ph0Lead}, \ref{line:ph0Echo}, \ref{line:RBCECCEncindicator}, \ref{line:RBCph2SI1}, \ref{line:RBCph2B}, \ref{line:RBCph2CSI0}, \ref{line:RBCph2ASI1}, \ref{line:RBCph2BSI0}, \ref{line:BRCReadySendA}, \ref{line:BRCRelialbeBeginBB}, \ref{line:BRCReadySendC}, \ref{line:BRCph3SendCorrectSymbols}. Specifically, in each communication, the node sends coded symbols or binary information to other nodes, where each symbol  carries only $\cb$ bits, for  $\cb =   \bigl \lceil \frac{ \max\{ \ell, \ k\cdot \log (n+1) \} }{k} \bigr\rceil$ and  $k =    \bigl \lfloor   \frac{ t  }{5 } \bigr\rfloor    +1$.  Also, the total number of communication steps for each node is finite.  Therefore, the communication per node is  $O(\max\{\Lh, n \log n \})$  bits, while the  total communication complexity of $\OciorRBC$ is  $O(\max\{n\Lh, n^2 \log n \})$  bits. 

$\OciorRBC$ consists of an initial phase and Phases~1-3.   The number of asynchronous rounds in these phases are:   $2$ rounds (see Lines~\ref{line:ph0Lead}, \ref{line:ph0Echo}),    $2$ rounds  (see Lines~\ref{line:RBCECCEncindicator}, \ref{line:RBCph2SI1}, \ref{line:RBCph2B}),  $2$ round  (see Lines~\ref{line:RBCph2CSI0}, \ref{line:RBCph2ASI1}, \ref{line:RBCph2BSI0},  \ref{line:BRCReadySendA}, \ref{line:BRCRelialbeBeginBB}, \ref{line:BRCReadySendC}), and $1$ round  (see Line~\ref{line:BRCph3SendCorrectSymbols}),  respectively.  
Therefore,   the round  complexity of $\OciorRBC$ is  $7$  rounds in the worst case.  
\end{proof}

\subsection{Lemmas used in the proof of Lemma~\ref{lm:BRCuniquegroup}}    \label{sec:ProofBRCuniquegroup}

The proof of Lemma~\ref{lm:BRCuniquegroup} will use the result of \cite[Lemma 8]{Chen:2020arxiv}.  
This result  considers a graph $G=(\Pc, \Ec)$, where $\Pc$ is a set of  $n -t$  vertices,  for $\Pc= [1: n-t]$ without loss of generality,  and $\Ec$ is a set of edges. 
In this graph,  there is a  given vertex $i^{\star}$ for $i^{\star}\in \Pc$, and a set of vertices $\Cc$ for $\Cc \subseteq \Pc \setminus \{i^{\star}\}$ and  $|\Cc|  \geq n-2t -1$,  such that each vertex in $\Cc$ is connected with at least $n-2t$ edges and that one of the edges is connected to vertex $i^{\star}$.  
Let $E_{i,j}= 1$ (respectively, $E_{i,j}= 0$)   denote the presence (respectively, absence)  of an edge  between vertex $i$ and vertex $j$, for $E_{i,j}= E_{j,i}, \forall i, j \in \Pc$.  
Mathematically, for this graph $G=(\Pc, \Ec)$, there exists a set  $\Cc \subseteq \Pc \setminus \{i^{\star}\}$ satisfying  the following conditions: 
 \begin{align}
E_{i, i^{\star}}&= 1                          \quad     \forall i \in \Cc   \label{eq:graph01}  \\
\sum_{j\in \Pc}E_{i, j} &\geq   n-2t   \quad     \forall i \in \Cc  \label{eq:graph02}   \\
|\Cc| &   \geq n-2t -1     \label{eq:graph03} 
 \end{align}
for a given $i^{\star}\in \Pc= [1: n-t]$.  
For this graph, we use $\Dc \subseteq \Pc$ to define  a set of vertices such that each vertex in $\Dc$ is connected with at least $k$ vertices in $\Cc$, i.e.,  
 \begin{align}
\Dc \defeq \Bigl\{i: \ \sum_{j\in \Cc }E_{i, j}  \geq  k , \   i \in  \Pc \setminus \{i^{\star}\} \Bigr\}  \label{eq:graphD01}   
 \end{align}
where $k$ is a parameter of  $(n,k)$ error correction code, which is set here as $k=       \lfloor   \frac{ t  }{5 }  \rfloor    +1$.  
For this  graph $G=(\Pc, \Ec)$,  the size of $\Dc$ can be bounded, based on the result of  \cite[Lemma 8]{Chen:2020arxiv} that is restated below.

\begin{lemma}  \cite[Lemma 8]{Chen:2020arxiv}  \label{lm:graph}
For any graph $G=(\Pc, \Ec)$ specified by \eqref{eq:graph01}-\eqref{eq:graph03} and for the set $\Dc \subseteq \Pc$ defined by \eqref{eq:graphD01}, and given $n\geq 3t+1$, the size of $\Dc$ is bounded as: 
 \begin{align}
|\Dc| &\geq  n-9t/4-1 .    \label{eq:graphDr11}  
 \end{align}
\end{lemma}

The proof of Lemma~\ref{lm:BRCuniquegroup} will also use the  following lemma, which  is obtained  by  following the  proof  of \cite[Lemma~9]{Chen:2020arxiv}. 

\begin{lemma}    \label{lm:ociorRBCsizem}
When $\gn^{[2]}\geq 1$, it holds true that $|\Ac_{\Lin}| \geq n-9t/4$, for any  $\Lin \in [1: \gn^{[2]}]$.
\end{lemma}
 \begin{proof}
By  following the  proof  of \cite[Lemma 9]{Chen:2020arxiv}, the proof here includes  the following key steps: 
 \begin{itemize}
\item  Step (a): Transform the network into a graph that is within the family of graphs satisfying \eqref{eq:graph01}-\eqref{eq:graph03}, for a fixed $i^{\star}$ in  $\Ac_{\Lin^{\star}}^{[2]}$ and $\Lin^{\star} \in [1: \gn^{[2]}]$.
\item  Step (b): Bound the size of a group of honest nodes, denoted by $\Dc'$ (with the same form as  in \eqref{eq:graphD01}), using the result of Lemma~\ref{lm:graph}, i.e., $|\Dc'| \geq n-9t/4-1$.
\item Step (c): Argue that every $\Pro$ in $\Dc'$ has the same initial message as $\PRO$~$i^{\star}$.
\item  Step (d): Conclude from Step (c) that $\Dc'$ is a subset of $\Ac_{\Lin^{\star}}$, i.e., $\Dc'\cup \{i^{\star}\} \subseteq \Ac_{\Lin^{\star}}$ and conclude that the size of $\Ac_{\Lin^{\star}}$ is bounded by the number determined in Step (b), i.e., $|\Ac_{\Lin^{\star}}| \geq |\Dc'| +1 \geq n-9t/4-1+1$, for  $\Lin^{\star} \in [1: \gn^{[2]}]$.
  \end{itemize}

 \emph{Step (a):} The first step of the proof is to transform the network into a graph that is within the family of graphs defined by \eqref{eq:graph01}-\eqref{eq:graph03}. We will consider the case of $\gn^{[2]}\geq 1$. 
 Let us consider a fixed $i^{\star}$ for $i^{\star}\in \Ac_{\Lin^{\star}}^{[2]}$ and $\Lin^{\star} \in [1: \gn^{[2]}]$, and given $\gn^{[2]}\geq 1$. 
 Based on the definition in  \eqref{eq:RBCAell}, in this setting  the honest Node~$i^{\star}$ sets the success indicator $\Ry_{i^{\star}}^{[2]} =1$  at   Phase~2, under the condition of   
   \begin{align}
   |\Ss_1^{[1]}|\geq  n- t    \label{eq:sindicatorSs} 
  \end{align}
(see Lines~\ref{line:RBCph2ACond} and \ref{line:RBCph2ASI1}). The condition in \eqref{eq:sindicatorSs}   implies the following inequalities: 
   \begin{align}
     |\Ss_1^{[1]}\cap \{\Fc \cup \{\cup_{p=1}^{\gn^{[1]}}\Ac_{p}^{[1]} \}\}| &\geq  n- t    \label{eq:sindicatorSsA}    \\
     |\Ss_1^{[1]}\cap  \{\cup_{p=1}^{\gn^{[1]}}\Ac_{p}^{[1]} \}| &\geq  n- t-t    \label{eq:sindicatorSsB}    \\
     |\Lkset_1 \cap  \{\cup_{p=1}^{\gn^{[1]}}\Ac_{p}^{[1]} \}| &\geq  n- t-t    \label{eq:sindicatorUs} 
  \end{align}
where $\Ss_1^{[1]}$ and $\Lkset_1$ are viewed from Node~$i^{\star}$;  \eqref{eq:sindicatorSsA} follows from the facts that $\Ry_j^{[1]} =0$ and that $j\notin \Ss_1^{[1]}$, $\forall j   \in [1:n] \setminus  \{\Fc \cup \{\cup_{p=1}^{\gn^{[1]}}\Ac_{p}^{[1]} \}\}$; \eqref{eq:sindicatorSsB} stems from the assumption that $|\Fc|\leq t$; and 
 \eqref{eq:sindicatorUs} is true due to the identity that   $\Ss_1^{[1]} \subseteq \Lkset_1$ (see Lines~\ref{line:RBCph1SS1Cond} and \ref{line:RBCph1SS1}).
  The condition in \eqref{eq:sindicatorUs} also implies that
    \begin{align}
\sum_{j \in \{\Lkset_1 \cap\{\cup_{p=1}^{\gn^{[1]}}\Ac_{p}^{[1]}\}\}\setminus \{i^{\star}\}} \Lk_{i^{\star}}^{[1]} (j)  \geq n - t -t  -1.    \label{eq:sindicator3lm05} 
  \end{align}  
  where $\Lk_{i^{\star}}^{[1]} (j)$ is the link indicator at   Phase~1 defined in \eqref{eq:lkindicator}.
Based on the  definition in  \eqref{eq:lkindicator}, it is true that  $\Lk_{i^{\star}}^{[1]} (j)  =1, \forall  j\in \Lkset_1$ (see Lines~\ref{line:RBCph1MatchCond} and \ref{line:RBCph1Match}) and that $\Lk_{i^{\star}}^{[1]} (j)  =1, \forall  j\in \Lkset_1 \cap  \{\cup_{p=1}^{\gn^{[1]}}\Ac_{p}^{[1]} \}$.

 For $i^{\star}\in \Ac_{\Lin^{\star}}^{[2]}$ and $\Lin^{\star} \in [1: \gn^{[2]}]$, let us define a subset of $\{\cup_{p=1}^{\gn^{[1]}}\Ac_{p}^{[1]}\}\setminus \{i^{\star}\}$ of honest nodes as 
     \begin{align}
\Cc' \defeq  \{j:    \Lk_{i^{\star}}^{[1]} (j) =1, j \in \{\cup_{p=1}^{\gn^{[1]}}\Ac_{p}^{[1]}\}\setminus \{i^{\star}\}  \} .    \label{eq:cdef01} 
  \end{align}
We can understand $\Cc'$  as a subset of $\{\cup_{p=1}^{\gn^{[1]}}\Ac_{p}^{[1]}\}\setminus \{i^{\star}\}$ of  honest nodes, in which each node sends a matched observation to Node~$i^{\star}$.  The observation sent from Node~$j$ to Node~$i^{\star}$ is defined by   $(y_{i^{\star}}^{(j)}, y_j^{(j)})$  (see Line~\ref{line:RBCECCEncindicator}). 
This observation is said to be matched if  $ (y_{i^{\star}}^{(j)}, y_j^{(j)}) = (y_{i^{\star}}^{(i^{\star})}, y_j^{(i^{\star})})$. 
One can see that  $\{\Lkset_1 \cap\{\cup_{p=1}^{\gn^{[1]}}\Ac_{p}^{[1]}\}\}\setminus \{i^{\star}\}$ is a subset of $\Cc'$ due to the fact that $\Lk_{i^{\star}}^{[1]} (j)  =1, \forall  j\in \Lkset_1 \cap  \{\cup_{p=1}^{\gn^{[1]}}\Ac_{p}^{[1]} \}$.  
 Based on \eqref{eq:sindicator3lm05} and \eqref{eq:cdef01}, the following conclusions are true  
 \begin{align}
\Lk_{j}^{[1]} (i^{\star}) &=1,  \quad \forall j  \in  \Cc'  \label{eq:cdef02}  \\
|\Cc'| &\geq     n - 2t  -1.  \label{eq:cdef03} 
 \end{align}
Note that  in our setting it holds true that $ \Lk_{i}^{[1]} (j)  =  \Lk_{j}^{[1]} (i)$, $\forall i, j  \in  \cup_{\Lin=1}^{\gn}\Ac_{\Lin}$ (see \eqref{eq:lkindicator}).

 Due to the fact that $\Cc' \subseteq \cup_{p=1}^{\gn^{[1]}}\Ac_{p}^{[1]}$ and  that  $ \Ry_j^{[1]} =1,   \forall j \in  \cup_{p=1}^{\gn^{[1]}}\Ac_{p}^{[1]}$ (see \eqref{eq:RBCAell}), the following conclusion is  true:    
  \begin{align}
\Ry_j^{[1]} =1, \quad  \forall j \in  \Cc' .   \label{eq:cdef04} 
 \end{align}
 The conclusion in \eqref{eq:cdef04} also implies   that 
   \begin{align}
\sum_{p =1}^n \Lk_{j}^{[1]} (p)  \geq n - t ,   \quad \forall j    \in  \Cc'   \label{eq:sindicator3lm0211} 
  \end{align}           
(see Line~\ref{line:RBCph2OneCond} and Line~\ref{line:RBCph2SI1}) and that  
   \begin{align}
\sum_{p  \in  \cup_{\Lin=1}^{\gn}\Ac_{\Lin}}  \Lk_{j}^{[1]} (p)  \geq n - 2t ,   \quad \forall j    \in  \Cc'   \label{eq:sindicator3lm0222} 
  \end{align}
where $\cup_{\Lin=1}^{\gn}\Ac_{\Lin}=[1:n] \setminus \Fc$.  Intuitively, for any $j    \in  \Cc'$, Node~$j$ receives at least $n-2t$ number of  matched observations  from the honest nodes  within $\cup_{\Lin=1}^{\gn}\Ac_{\Lin}$.  Let us define a subset of $\{\cup_{\Lin=1}^{\gn}\Ac_{\Lin}\} \setminus \{i^{\star}\}$ of honest $\Pros$ as 
 \begin{align}
\Dc' \defeq \Bigl\{p: \   \sum_{j    \in  \Cc' }  \Lk_{j}^{[1]} (p)    \geq  k , \   p  \in  \{\cup_{\Lin=1}^{\gn}\Ac_{\Lin}\} \setminus \{i^{\star}\} \Bigr\}  \label{eq:graphD01prime}   
 \end{align}
where $k$ is set as $k=       \lfloor   \frac{ t  }{5 }  \rfloor    +1$.  
 We can understand  $\Dc'$  as a set of  honest nodes in which each node sends at least $k$ matched observations to the nodes in $ \Cc'$. 
 
At this point,  we map the network into a graph $G=(\Pc', \Ec')$, where   $\Pc'$ is a set of $n -t$  vertices defined by $\Pc' = \cup_{\Lin=1}^{\gn}\Ac_{\Lin}$, and $\Ec'$ is a set of edges defined by $E_{i,j}= \Lk_{i}^{[1]} (j), \forall i, j  \in \Pc'$. 
For this graph $G=(\Pc', \Ec')$,  there exists a set  $\Cc' \subseteq \Pc' \setminus \{i^{\star}\}$ such that the  conditions in  \eqref{eq:cdef02}, \eqref{eq:cdef03} and \eqref{eq:sindicator3lm0222} are satisfied, for a given $i^{\star}\in \Ac_{\Lin^{\star}}^{[2]} \subseteq \Pc'$.    Since conditions in  \eqref{eq:cdef02}, \eqref{eq:cdef03} and \eqref{eq:sindicator3lm0222} are similar to the conditions in \eqref{eq:graph01}, \eqref{eq:graph03} and \eqref{eq:graph02}, respectively, this  graph $G=(\Pc', \Ec')$ falls into a family of graphs satisfying \eqref{eq:graph01}-\eqref{eq:graph03}.

\emph{Steps (b)-(d):} The remaining steps of this proof are similar to the steps (b)-(d) of the proof of \cite[Lemma 9]{Chen:2020arxiv}. 
\end{proof}

\begin{lemma}  \label{lm:BRCeta2bound2}
For the proposed  $\OciorRBC$ with $n\geq 3t+1$, it holds true that $\gn^{[2]} \leq 2$.
\end{lemma}
 \begin{proof}
 This proof is based on the result  of Lemma~\ref{lm:ociorRBCsizem}.  
The proof of this lemma is similar to that of  \cite[Lemma 11]{Chen:2020arxiv}.  The details are omitted here. 
 \end{proof}

\begin{lemma}    \label{lm:BRCeta22212}
For the proposed  $\OciorRBC$ with $n\geq 3t+1$,  if $\gn^{[2]} = 2$,  then it holds true that $\gn^{[1]} =   2$.   
\end{lemma}
\begin{proof}
 This proof is based on the result  of Lemma~\ref{lm:sizeboundMatrix}    and Lemma~\ref{lm:ociorRBCsizem}.  
The proof of this lemma is similar to that of   \cite[Lemma 13]{Chen:2020arxiv}.  The details are omitted here. 
\end{proof}

\begin{lemma}    \label{lm:BRCeta1231}
For the proposed  $\OciorRBC$ with $n\geq 3t + 1$, if $\gn^{[1]} = 2$  then it holds true that $\gn^{[2]} \leq   1$.
\end{lemma}
\begin{proof}
The proof of this lemma is similar to that of  Lemma~\ref{lm:eta1231} of $\OciorCOOL$.   
We will consider the assumption of $\gn^{[1]}= 2$. Under this assumption,  the definition in \eqref{eq:RBCAell00}-\eqref{eq:RBCBdef01} implies that 
 \begin{align}
 \Ac_{1}^{[1]} =&   \{  i:  \Ry_i^{[1]} =1,\Me_i^{[0]}  =  \Meg_{1},  \  i \notin  \Fc, \ i \in [1:n]\}       \label{eq:Aellbig10455RBC}  \\ 
  \Ac_{2}^{[1]} =&   \{  i:  \Ry_i^{[1]} =1, \Me_i^{[0]}  =  \Meg_{2},  \  i \notin  \Fc, \ i \in [1:n]\}      \label{eq:Aellbig1045522RBC}    \\ 
\Ac_{1,2}^{[1]} =   & \{  i:  \   i\in  \Ac_{1}^{[1]},   \   \hv_i^\T  \Meg_{1}  = \hv_i^\T  \Meg_2 \}    \label{eq:Alj2995RBC}  \\
\Ac_{1,1}^{[1]} =   & \Ac_{1}^{[1]} \setminus \Ac_{1,2}^{[1]} = \{  i:  \   i\in  \Ac_{1}^{[1]},   \   \hv_i^\T  \Meg_{1}  \neq  \hv_i^\T  \Meg_2 \}  \label{eq:All2955251RBC} \\
\Ac_{2,1}^{[1]} =   & \{  i:  \   i\in  \Ac_{2}^{[1]},   \   \hv_i^\T  \Meg_{2}  = \hv_i^\T  \Meg_1 \}   \label{eq:Alj299535RBC}  \\
\Ac_{2,2}^{[1]} =   & \Ac_{2}^{[1]} \setminus \Ac_{2,1}^{[1]} = \{  i:  \   i\in  \Ac_{2}^{[1]},   \   \hv_i^\T  \Meg_{2}  \neq  \hv_i^\T  \Meg_1 \}  \label{eq:All29559386RBC}  \\
\Bc^{[1]} =  & \{  i:  \Ry_i^{[1]} =0  ~\text{or $\Ry_i^{[1]}$ has never been set}, \  i \notin  \Fc, \ i \in [1:n] \}    \label{eq:BdefB1RBC}         
 \end{align} 
 where  $\Me_{i}^{[0]}$  denotes the value of $\Me_{i}$ updated at Phase~$0$. 
If Node~$i$ never updates  the value of $\Me_{i}$ before termination, then  $\Me_{i}^{[0]}$ is considered to be a default value $\Me_{i}^{[0]}=  \defaultvalue$.  
Since $|\Ac_{1}|+|\Ac_{2}|= n-|\Fc|-\sum_{\Lin=3}^{\gn} |\Ac_{\Lin}|$,  it is true that at least one of the following cases is satisfied: 
\begin{align}
\text{Case~1:} \quad  
|\Ac_{2}| \leq  &  \frac{n-|\Fc|-\sum_{\Lin=3}^{\gn} |\Ac_{\Lin}|}{2}   \label{eq:RBCcase1b} 
 \end{align}
\begin{align}
\text{Case~2:} \quad |\Ac_{1}| \leq &\frac{n-|\Fc|-\sum_{\Lin=3}^{\gn} |\Ac_{\Lin}|}{2}  .   \label{eq:RBCcase2a} 
 \end{align}

\subsubsection{Analysis for Case~1}   
We will first consider Case~1 and prove that $|\Ac_{2}^{[2]}| =0$ under this case. 
Let us define $\Uc_{i}^{[p]}$ as a set of links that are matched with Node~$i$ at Phase~$p$, that is, 
\begin{align}
\Uc_{i}^{[p]} :=  & \{  j:  \Lk_{i}^{[p]} (j)  =1, j \in [1:n]\},  \quad  \text{for}  \quad i\in [1:n],  p\in \{1,2\},      
\end{align}
where $\Lk_{i}^{[p]} (j) $ is defined in \eqref{eq:lkindicator} and \eqref{eq:lkindicator2}. 
 Then, for any $i\in \Ac_{2}^{[1]}$, the size of $\Uc_{i}^{[2]}$ can be bounded as 
 \begin{align}  
 |\Uc_{i}^{[2]}| = & \sum_{ j \in [1:n]}   \Lk_{i}^{[2]} (j)        \label{eq:link124RBC} \\
 = & \sum_{ j \in [1:n]\setminus \Bc^{[1]}}   \Lk_{i}^{[2]} (j)        \label{eq:link4252RBC} \\
  = & \sum_{ j \in [1:n]\setminus \{\Bc^{[1]}\cup \Ac_{1,1}^{[1]}\}}   \Lk_{i}^{[2]} (j)        \label{eq:link82525RBC} \\
    = & \sum_{ j \in  \{\Ac_{1,2}^{[1]}\cup \Ac_{2,1}^{[1]}\cup \Ac_{2,2}^{[1]} \cup\Fc \}}   \Lk_{i}^{[2]} (j)        \label{eq:link7256RBC} \\
  \leq & |\Ac_{1,2}^{[1]}|  + |\Ac_{2,1}^{[1]}|  + | \Ac_{2,2}^{[1]}|+|\Fc|              \label{eq:link6254RBC} 
\end{align}
where \eqref{eq:link4252RBC} stems from the fact that $\Ry_j^{[1]} =0$ for any $j\in \Bc^{[1]}$ (see \eqref{eq:BdefB1RBC}), which suggests that $\Lk_{i}^{[2]} (j)  =0$ (see \eqref{eq:lkindicator2}); 
\eqref{eq:link82525RBC} results from the identity that  $\hv_j^\T  \Meg_{1}  \neq  \hv_j^\T  \Meg_2 $ for any  $j\in \Ac_{1,1}^{[1]}$  (see \eqref{eq:All2955251RBC}), which implies that   $\Lk_{i}^{[1]} (j)  =0$;
\eqref{eq:link7256RBC} is from the fact that $[1:n]=   \Ac_{1,1}^{[1]} \cup \Ac_{1,2}^{[1]} \cup\Ac_{2,1}^{[1]}\cup \Ac_{2,2}^{[1]}\cup \Bc^{[1]} \cup \Fc$.

From Lemma~\ref{lm:sizeboundMatrix}, the summation of $|\Ac_{1,2}^{[1]}|  + |\Ac_{2,1}^{[1]}|$ in \eqref{eq:link6254RBC} can be bounded as 
 \begin{align}  
   |\Ac_{1,2}^{[1]}|  + |\Ac_{2,1}^{[1]}|   \leq k-1  .         \label{eq:link3235252RBC} 
\end{align}
From Lemma~\ref{lm:boundA22}, the term $|\Ac_{2,2}^{[1]}|$ in \eqref{eq:link6254RBC} can be upper bounded by 
 \begin{align}
|\Ac_{2,2}^{[1]}| \leq  2(k-1).   \label{eq:A22bound23435RBC} 
\end{align}
At this point, for any $i\in \Ac_{2}^{[1]}$, the size of $\Uc_{i}^{[2]}$ can be bounded as 
 \begin{align}  
 |\Uc_{i}^{[2]}|   \leq & \  |\Ac_{1,2}^{[1]}|  + |\Ac_{2,1}^{[1]}|  + | \Ac_{2,2}^{[1]}|+|\Fc|    \label{eq:A22bound1133RBC} \\
   \leq &\  k-1  + 2(k-1)       +|\Fc|   \label{eq:A22bound8255RBC}  \\ 
  \leq & \ 3(k- 1)  + t   \non  \\ 
    = & \ 3( \lfloor    t /5   \rfloor    +1 -1)   + t  \label{eq:A22bound65464RBC}  \\ 
    \leq &\  2 t   \non  \\ 
     <  & \  n-t  \label{eq:A22bound5366RBC}  
\end{align} 
where  \eqref{eq:A22bound1133RBC} is from    \eqref{eq:link6254RBC};   \eqref{eq:A22bound8255RBC} is from Lemma~\ref{lm:sizeboundMatrix} and Lemma~\ref{lm:boundA22}; \eqref{eq:A22bound65464RBC} uses the value of the parameter $k$, i.e., $k= \lfloor    t /5   \rfloor    +1$; and the last inequality uses the condition of $n\geq 3t+1>3t$.  
With the result in \eqref{eq:A22bound5366RBC}, it is concluded that    $\Ry_i^{[2]}$ cannot be $1$ for any $i\in \Ac_{2}^{[1]}$,  due to the derived result $|\Uc_{i}^{[2]}| <n-t$ (see Lines~\ref{line:RBCph2ACond} and \ref{line:RBCph2ASI1}). 
Therefore, it can be concluded that 
  \begin{align}
|\Ac_{2}^{[2]}|  =0 \label{eq:case1A2S0RBC} 
\end{align}
  for Case~1.

\subsubsection{Analysis for Case~2}  
By interchange the roles of $\Ac_{1}$ and  $\Ac_{2}$, one can easily follow the proof for Case~1 and  show that 
   \begin{align}
|\Ac_{1}^{[2]}|  =0  \label{eq:case1A1S0RBC} 
\end{align}
for Case~2. 
Then it completes the proof of this lemma.  
\end{proof}


\begin{thebibliography}{10}
\providecommand{\url}[1]{#1}
\csname url@samestyle\endcsname
\providecommand{\newblock}{\relax}
\providecommand{\bibinfo}[2]{#2}
\providecommand{\BIBentrySTDinterwordspacing}{\spaceskip=0pt\relax}
\providecommand{\BIBentryALTinterwordstretchfactor}{4}
\providecommand{\BIBentryALTinterwordspacing}{\spaceskip=\fontdimen2\font plus
\BIBentryALTinterwordstretchfactor\fontdimen3\font minus
  \fontdimen4\font\relax}
\providecommand{\BIBforeignlanguage}[2]{{%
\expandafter\ifx\csname l@#1\endcsname\relax
\typeout{** WARNING: IEEEtran.bst: No hyphenation pattern has been}%
\typeout{** loaded for the language `#1'. Using the pattern for}%
\typeout{** the default language instead.}%
\else
\language=\csname l@#1\endcsname
\fi
#2}}
\providecommand{\BIBdecl}{\relax}
\BIBdecl

\bibitem{PSL:80}
M.~Pease, R.~Shostak, and L.~Lamport, ``Reaching agreement in the presence of
  faults,'' \emph{Journal of the ACM}, vol.~27, no.~2, pp. 228--234, Apr. 1980.

\bibitem{LSP:82}
L.~Lamport, R.~Shostak, and M.~Pease, ``The {Byzantine} generals problem,''
  \emph{ACM Transactions on Programming Languages and Systems (TOPLAS)},
  vol.~4, no.~3, pp. 382--401, Jul. 1982.

\bibitem{ChenDISC:21}
J.~Chen, ``Optimal error-free multi-valued {Byzantine} agreement,'' in
  \emph{International Symposium on Distributed Computing (DISC)}, Oct. 2021.

\bibitem{Chen:2020arxiv}
------, ``Fundamental limits of {Byzantine} agreement,'' 2020, available on
  ArXiv: https://arxiv.org/pdf/2009.10965.pdf.

\bibitem{LCabaISIT:21}
F.~Li and J.~Chen, ``Communication-efficient signature-free asynchronous
  {B}yzantine agreement,'' in \emph{Proc. {IEEE} Int. Symp. Inf. Theory
  {(ISIT)}}, Jul. 2021.

\bibitem{ZLC:23}
J.~Zhu, F.~Li, and J.~Chen, ``Communication-efficient and error-free gradecast
  with optimal resilience,'' in \emph{Proc. {IEEE} Int. Symp. Inf. Theory
  {(ISIT)}}, Jun. 2023, pp. 108--113.

\bibitem{FH:06}
M.~Fitzi and M.~Hirt, ``Optimally efficient multi-valued {Byzantine}
  agreement,'' in \emph{Proceedings of the ACM Symposium on Principles of
  Distributed Computing (PODC)}, Jul. 2006, pp. 163--168.

\bibitem{LV:11}
G.~Liang and N.~Vaidya, ``Error-free multi-valued consensus with {Byzantine}
  failures,'' in \emph{Proceedings of the ACM Symposium on Principles of
  Distributed Computing (PODC)}, Jun. 2011, pp. 11--20.

\bibitem{GP:20}
C.~Ganesh and A.~Patra, ``Optimal extension protocols for {Byzantine} broadcast
  and agreement,'' in \emph{Distributed Computing}, Jul. 2020.

\bibitem{LDK:20}
A.~Loveless, R.~Dreslinski, and B.~Kasikci, ``Optimal and error-free
  multi-valued {Byzantine} consensus through parallel execution,'' 2020,
  available on : https://eprint.iacr.org/2020/322.

\bibitem{NRSVX:20}
K.~Nayak, L.~Ren, E.~Shi, N.~Vaidya, and Z.~Xiang, ``Improved extension
  protocols for {Byzantine} broadcast and agreement,'' in \emph{International
  Symposium on Distributed Computing (DISC)}, Oct. 2020.

\bibitem{Patra:11}
A.~Patra, ``Error-free multi-valued broadcast and {Byzantine} agreement with
  optimal communication complexity,'' in \emph{International Conference on
  Principles of Distributed Systems (OPODIS)}, 2011, pp. 34--49.

\bibitem{CT:05}
C.~Cachin and S.~Tessaro, ``Asynchronous verifiable information dispersal,'' in
  \emph{IEEE Symposium on Reliable Distributed Systems (SRDS)}, Oct. 2005.

\bibitem{CDG+:24}
P.~Civit, M.~A. Dzulfikar, S.~Gilbert, R.~Guerraoui, J.~Komatovic,
  M.~Vidigueira, and I.~Zablotchi, ``Efficient signature-free validated
  agreement,'' 2024, arXiv:2403.08374.

\bibitem{ADD+:22}
N.~Alhaddad, S.~Das, S.~Duan, L.~Ren, M.~Varia, Z.~Xiang, and H.~Zhang,
  ``Balanced {Byzantine} reliable broadcast with near-optimal communication and
  improved computation,'' in \emph{Proceedings of the ACM Symposium on
  Principles of Distributed Computing (PODC)}, Jul. 2022, pp. 399--417.

\bibitem{BGP:92}
P.~Berman, J.~Garay, and K.~Perry, ``Bit optimal distributed consensus,''
  \emph{Computer Science}, pp. 313--321, 1992.

\bibitem{CW:92}
B.~Coan and J.~Welch, ``Modular construction of a {Byzantine} agreement
  protocol with optimal message bit complexity,'' \emph{Information and
  Computation}, vol.~97, no.~1, pp. 61--85, Mar. 1992.

\bibitem{Bracha:87}
G.~Bracha, ``Asynchronous {Byzantine} agreement protocols,'' \emph{Information
  and Computation}, vol.~75, no.~2, pp. 130--143, Nov. 1987.

\bibitem{RS:60}
I.~Reed and G.~Solomon, ``Polynomial codes over certain finite fields,''
  \emph{Journal of the Society for Industrial and Applied Mathematics}, vol.~8,
  no.~2, pp. 300--304, Jun. 1960.

\bibitem{roth:06}
R.~Roth, \emph{Introduction to coding theory}.\hskip 1em plus 0.5em minus
  0.4em\relax Cambridge University Press, 2006.

\bibitem{Berlekamp:68}
E.~Berlekamp, ``Nonbinary {BCH} decoding (abstr.),'' \emph{IEEE Trans. Inf.
  Theory}, vol.~14, no.~2, pp. 242--242, Mar. 1968.

\bibitem{BCG:93}
M.~Ben-Or, R.~Canetti, and O.~Goldreich, ``Asynchronous secure computation,''
  in \emph{Proceedings of the Twenty-Fifth Annual ACM Symposium on Theory of
  Computing}, 1993, pp. 52--61.

\end{thebibliography}


\end{document}